\documentclass[10pt,showpacs,a4paper,aps,pra,onecolumn,nofootinbib,superscriptaddress,floatfix]{revtex4-2}
\pdfoutput=1

\usepackage{rotating}
\usepackage{floatpag}
\rotfloatpagestyle{empty}
\renewcommand{\theenumi}{\arabic{enumi}}
\usepackage{dsfont}
\usepackage{amsmath,amsthm,amssymb}
\usepackage{xspace,enumerate,epsfig}
\usepackage{graphicx}
\graphicspath{{.}{./figures/}}
\usepackage{marvosym}

\usepackage{tikzfig}
\usepackage{stmaryrd}
\usepackage{docmute}
\usepackage{keycommand}

\newkeycommand{\pointmap}[style=point][1]{\,\tikz{\node[style=\commandkey{style}] (x) at (0,-0.3) {$#1$};\node [style=none] (3) at (0, 1.0) {};\node [style=none] (3) at (0, -1.0) {};\draw (x)--(0,0.7);}\,}

\newkeycommand{\typedkpoint}[style=kpoint,edgestyle=][2]{%
\,\begin{tikzpicture}
  \begin{pgfonlayer}{nodelayer}
    \node [style=none] (0) at (0, 1) {};
    \node [style=\commandkey{style}] (1) at (0, -0.75) {$#2$};
    \node [style=right label] (2) at (0.25, 0.5) {$#1$};
  \end{pgfonlayer}
  \begin{pgfonlayer}{edgelayer}
    \draw [\commandkey{edgestyle}] (1) to (0.center);
  \end{pgfonlayer}
\end{tikzpicture}\,}

\newkeycommand{\typedkpointdag}[style=kpoint adjoint,edgestyle=][2]{%
\,\begin{tikzpicture}
  \begin{pgfonlayer}{nodelayer}
    \node [style=none] (0) at (0, -1) {};
    \node [style=\commandkey{style}] (1) at (0, 0.75) {$#2$};
    \node [style=right label] (2) at (0.25, -0.5) {$#1$};
  \end{pgfonlayer}
  \begin{pgfonlayer}{edgelayer}
    \draw [\commandkey{edgestyle}] (0.center) to (1);
  \end{pgfonlayer}
\end{tikzpicture}\,}

\newcommand{\bistateadj}[1]{\,\begin{tikzpicture}[yshift=-3mm]
  \begin{pgfonlayer}{nodelayer}
    \node [style=kpointadj, minimum width=1 cm, inner sep=2pt] (0) at (0, 0.5) {$\psi$};
    \node [style=none] (1) at (-0.75, 0.25) {};
    \node [style=none] (2) at (0.75, 0.25) {};
    \node [style=none] (3) at (-0.75, -0.5) {};
    \node [style=none] (4) at (0.75, -0.5) {};
  \end{pgfonlayer}
  \begin{pgfonlayer}{edgelayer}
    \draw (1.center) to (3.center);
    \draw (2.center) to (4.center);
  \end{pgfonlayer}
\end{tikzpicture}\,}

\newkeycommand{\pointketbra}[style1=copoint,style2=point][2]{\,%
\begin{tikzpicture}
  \begin{pgfonlayer}{nodelayer}
    \node [style=none] (0) at (0, -1.5) {};
    \node [style=\commandkey{style1}] (1) at (0, -0.75) {$#1$};
    \node [style=\commandkey{style2}] (2) at (0, 0.75) {$#2$};
    \node [style=none] (3) at (0, 1.5) {};
  \end{pgfonlayer}
  \begin{pgfonlayer}{edgelayer}
    \draw (0.center) to (1);
    \draw (2) to (3.center);
  \end{pgfonlayer}
\end{tikzpicture}\,}

\newkeycommand{\twocopointketbra}[style1=copoint,style2=copoint,style3=point][3]{\,%
\begin{tikzpicture}
  \begin{pgfonlayer}{nodelayer}
    \node [style=none] (0) at (-0.75, -1.5) {};
    \node [style=none] (1) at (0.75, -1.5) {};
    \node [style=\commandkey{style1}] (2) at (-0.75, -0.75) {$#1$};
    \node [style=\commandkey{style2}] (3) at (0.75, -0.75) {$#2$};
    \node [style=\commandkey{style3}] (4) at (0, 0.75) {$#3$};
    \node [style=none] (5) at (0, 1.5) {};
  \end{pgfonlayer}
  \begin{pgfonlayer}{edgelayer}
    \draw (0.center) to (2);
    \draw (1.center) to (3);
    \draw (4) to (5.center);
  \end{pgfonlayer}
\end{tikzpicture}\,}

\newkeycommand{\twopointketbra}[style1=copoint,style2=point,style3=point][3]{\,%
\begin{tikzpicture}
  \begin{pgfonlayer}{nodelayer}
    \node [style=none] (0) at (0, -1.5) {};
    \node [style=none] (1) at (-0.75, 1.5) {};
    \node [style=\commandkey{style1}] (2) at (0, -0.75) {$#1$};
    \node [style=\commandkey{style2}] (3) at (-0.75, 0.75) {$#2$};
    \node [style=\commandkey{style3}] (4) at (0.75, 0.75) {$#3$};
    \node [style=none] (5) at (0.75, 1.5) {};
  \end{pgfonlayer}
  \begin{pgfonlayer}{edgelayer}
    \draw (0.center) to (2);
    \draw (1.center) to (3);
    \draw (4) to (5.center);
  \end{pgfonlayer}
\end{tikzpicture}\,}

\newkeycommand{\pointbraket}[style1=point,style2=copoint][2]{\,%
\begin{tikzpicture}
  \begin{pgfonlayer}{nodelayer}
    \node [style=\commandkey{style1}] (0) at (0, -0.5) {$#1$};
    \node [style=\commandkey{style2}] (1) at (0, 0.5) {$#2$};
  \end{pgfonlayer}
  \begin{pgfonlayer}{edgelayer}
    \draw (0) to (1);
  \end{pgfonlayer}
\end{tikzpicture}\,}

\newkeycommand{\kpointbraket}[style1=kpoint,style2=kpoint adjoint][2]{\,%
\begin{tikzpicture}
  \begin{pgfonlayer}{nodelayer}
    \node [style=\commandkey{style1}] (0) at (0, -0.75) {$#1$};
    \node [style=\commandkey{style2}] (1) at (0, 0.75) {$#2$};
  \end{pgfonlayer}
  \begin{pgfonlayer}{edgelayer}
    \draw (0) to (1);
  \end{pgfonlayer}
\end{tikzpicture}\,}

\newkeycommand{\kpointmap}[style1=kpoint,style2=map][2]{\,%
\begin{tikzpicture}
  \begin{pgfonlayer}{nodelayer}
    \node [style=\commandkey{style1}] (0) at (0, -1.1) {$#1$};
    \node [style=\commandkey{style2}] (1) at (0, 0.2) {$#2$};
    \node [style=none] (2) at (0, 1.5) {};
  \end{pgfonlayer}
  \begin{pgfonlayer}{edgelayer}
    \draw (0) to (1);
    \draw (1) to (2.center);
  \end{pgfonlayer}
\end{tikzpicture}%
\,}

\newkeycommand{\sandwichtwo}[style1=point,style2=map,style3=map,style4=copoint][4]{\,%
\begin{tikzpicture}
  \begin{pgfonlayer}{nodelayer}
    \node [style=\commandkey{style2}] (0) at (0, -0.75) {$#2$};
    \node [style=\commandkey{style3}] (1) at (0, 0.75) {$#3$};
    \node [style=\commandkey{style1}] (2) at (0, -2) {$#1$};
    \node [style=\commandkey{style4}] (3) at (0, 2) {$#4$};
  \end{pgfonlayer}
  \begin{pgfonlayer}{edgelayer}
    \draw (2) to (0);
    \draw (0) to (1);
    \draw (1) to (3);
  \end{pgfonlayer}
\end{tikzpicture}\,}

\newkeycommand{\longbraket}[style1=point,style2=copoint][2]{\,%
\begin{tikzpicture}
  \begin{pgfonlayer}{nodelayer}
    \node [style=\commandkey{style1}] (0) at (0, -2) {$#1$};
    \node [style=\commandkey{style2}] (1) at (0, 2) {$#2$};
  \end{pgfonlayer}
  \begin{pgfonlayer}{edgelayer}
    \draw (0) to (1);
  \end{pgfonlayer}
\end{tikzpicture}\,}

\newkeycommand{\onb}[style=point]{\ensuremath{\left\{\tikz{\node[style=\commandkey{style}] (x) at (0,-0.3) {$j$};\draw (x)--(0,0.7);}\right\}}\xspace}

\newkeycommand{\copointmap}[style=copoint][1]{\,\tikz{\node[style=\commandkey{style}] (x) at (0,0.3) {$#1$};\draw (x)--(0,-0.7);}\,}

\tikzstyle{cdiag}=[matrix of math nodes, row sep=3em, column sep=3em, text height=1.5ex, text depth=0.25ex,inner sep=0.5em]
\tikzstyle{arrow above}=[transform canvas={yshift=0.5ex}]
\tikzstyle{arrow below}=[transform canvas={yshift=-0.5ex}]

\usetikzlibrary{shapes.multipart}

\tikzstyle{bwSpider}=[
       rectangle split,
       rectangle split parts=2,
       rectangle split part fill={black,white},
 minimum size=3.6 mm, inner sep=-2mm, draw=black,scale=0.5,rounded corners=0.8 mm
       ]
 \tikzstyle{wbSpider}=[
       rectangle split,
       rectangle split parts=2,
       rectangle split part fill={white,black},
 minimum size=3.6 mm, inner sep=-2mm, draw=black,scale=0.5,rounded corners=0.8 mm
       ]
\tikzstyle{qWire}=[line width = 1pt, color=quantumviolet]
\tikzstyle{pWire}=[line width = 1pt, color=red!60!black]
\tikzstyle{gWire}=[line width = 1.5pt, color=green!60!black]
\tikzstyle{cWire}=[color=gray,thin]
\tikzstyle{env}=[copoint,regular polygon rotate=0,minimum width=0.2cm, fill=black]

\tikzstyle{probs}=[shape=semicircle,fill=white,draw=black,shape border rotate=180,minimum width=1.2cm]

\tikzstyle{every picture}=[baseline=-0.25em,scale=0.5]
\tikzstyle{dotpic}=[] 
\tikzstyle{diredges}=[every to/.style={diredge}]
\tikzstyle{math matrix}=[matrix of math nodes,left delimiter=(,right delimiter=),inner sep=2pt,column sep=1em,row sep=0.5em,nodes={inner sep=0pt},text height=1.5ex, text depth=0.25ex]

\tikzstyle{inline text}=[text height=1.5ex, text depth=0.25ex,yshift=0.5mm]
\tikzstyle{label}=[font=\footnotesize,text height=1.5ex, text depth=0.25ex,yshift=0.5mm]
\tikzstyle{left label}=[label,anchor=east,xshift=1.5mm]
\tikzstyle{right label}=[label,anchor=west,xshift=-1mm]

\tikzstyle{braceedge}=[decorate,decoration={brace,amplitude=2mm,raise=-1mm}]
\tikzstyle{small braceedge}=[decorate,decoration={brace,amplitude=1mm,raise=-1mm}]

\tikzstyle{doubled}=[line width=1.6pt] 
\tikzstyle{boldedge}=[doubled,shorten <=-0.17mm,shorten >=-0.17mm]
\tikzstyle{boldedgegray}=[doubled,gray,shorten <=-0.17mm,shorten >=-0.17mm]
\tikzstyle{singleedgegray}=[gray]

\tikzstyle{semidoubled}=[line width=1.4pt] 
\tikzstyle{semiboldedgegray}=[semidoubled,gray,shorten <=-0.17mm,shorten >=-0.17mm]

\tikzstyle{boxedge}=[semiboldedgegray]

\tikzstyle{boldedgedashed}=[very thick,dashed,shorten <=-0.17mm,shorten >=-0.17mm]
\tikzstyle{vboldedgedashed}=[doubled,dashed,shorten <=-0.17mm,shorten >=-0.17mm]
\tikzstyle{left hook arrow}=[left hook-latex]
\tikzstyle{right hook arrow}=[right hook-latex]
\tikzstyle{sembracket}=[line width=0.5pt,shorten <=-0.07mm,shorten >=-0.07mm]

\tikzstyle{causal edge}=[->,thick,gray]
\tikzstyle{causal nondir}=[thick,gray]
\tikzstyle{timeline}=[thick,gray, dashed]

\tikzstyle{cedge}=[<->,thick,gray!70!white]

\tikzstyle{empty diagram}=[draw=gray!40!white,dashed,shape=rectangle,minimum width=1cm,minimum height=1cm]
\tikzstyle{empty diagram small}=[draw=gray!50!white,dashed,shape=rectangle,minimum width=0.6cm,minimum height=0.5cm]

\tikzstyle{dot}=[inner sep=0mm,minimum width=2mm,minimum height=2mm,draw,shape=circle]
\tikzstyle{leak}=[white dot, shape=regular polygon, minimum size=3.3 mm, regular polygon sides=3, outer sep=-0.2mm, regular polygon rotate=270]
\tikzstyle{proj}=[draw,fill=white,chamfered rectangle,chamfered rectangle angle=30, minimum width=2mm, minimum height= 1mm,scale=0.5, outer sep=-0.2mm]
\tikzstyle{wide proj}=[draw,fill=white,chamfered rectangle,chamfered rectangle angle=30, minimum width=15mm,minimum height=1mm,scale=0.5, outer sep=-0.2mm]
\tikzstyle{very wide proj}=[draw,fill=white,chamfered rectangle,chamfered rectangle angle=30, minimum width=25mm,minimum height=1mm,scale=0.5, outer sep=-0.2mm]
\tikzstyle{very very wide proj}=[draw,fill=white,chamfered rectangle,chamfered rectangle angle=30, minimum width=35mm,minimum height=1mm,scale=0.5, outer sep=-0.2mm]
\tikzstyle{preleak}=[proj]

\tikzstyle{split proj out}=[regular polygon,regular polygon sides=3,draw,scale=0.75,inner sep=-0.5pt,minimum width=3.3mm,fill=white,regular polygon rotate=180]
\tikzstyle{split proj in}=[regular polygon,regular polygon sides=3,draw,scale=0.75,inner sep=-0.5pt,minimum width=3.3mm,fill=white]
\tikzstyle{Vleak}=[white dot, shape=regular polygon, minimum size=3.3 mm, regular polygon sides=3, outer sep=-0.2mm, regular polygon rotate=90]
\tikzstyle{dleak}=[white dot, line width=1.6pt, shape=regular polygon, minimum size=3.3 mm, regular polygon sides=3, outer sep=-0.2mm, regular polygon rotate=270]

\tikzstyle{Wsquare}=[white dot, shape=regular polygon, rounded corners=0.8 mm, minimum size=3.3 mm, regular polygon sides=3, outer sep=-0.2mm]
\tikzstyle{Wsquareadj}=[white dot, shape=regular polygon, rounded corners=0.8 mm, minimum size=3.3 mm, regular polygon sides=3, outer sep=-0.2mm, regular polygon rotate=180]
\tikzstyle{ddot}=[inner sep=0mm, doubled, minimum width=2.5mm,minimum height=2.5mm,draw,shape=circle]

\tikzstyle{black dot}=[dot,fill=black]
\tikzstyle{white dot}=[dot,fill=white,,text depth=-0.2mm]
\tikzstyle{white Wsquare}=[Wsquare,fill=gray,,text depth=-0.2mm]
\tikzstyle{white Wsquareadj}=[Wsquareadj,fill=white,,text depth=-0.2mm]
\tikzstyle{green dot}=[white dot] 
\tikzstyle{gray dot}=[dot,fill=gray!40!white,,text depth=-0.2mm]
\tikzstyle{red dot}=[gray dot] 

\tikzstyle{black ddot}=[ddot,fill=black]
\tikzstyle{white ddot}=[ddot,fill=white]
\tikzstyle{gray ddot}=[ddot,fill=gray!40!white]

\tikzstyle{gray edge}=[gray!60!white]

\tikzstyle{small dot}=[inner sep=0.5mm,minimum width=0pt,minimum height=0pt,draw,shape=circle]

\tikzstyle{small black dot}=[small dot,fill=black]
\tikzstyle{small white dot}=[small dot,fill=white]
\tikzstyle{small gray dot}=[small dot,fill=gray!40!white]

\tikzstyle{causal dot}=[inner sep=0.4mm,minimum width=0pt,minimum height=0pt,draw=white,shape=circle,fill=gray!40!white]

\tikzstyle{phase dimensions}=[minimum size=5mm,font=\footnotesize,rectangle,rounded corners=2.5mm,inner sep=0.2mm,outer sep=-2mm]
\tikzstyle{dphase dimensions}=[minimum size=5mm,font=\footnotesize,rectangle,rounded corners=2.5mm,inner sep=0.2mm,outer sep=-2mm]

\tikzstyle{white phase dot}=[dot,fill=white,phase dimensions]
\tikzstyle{white phase ddot}=[ddot,fill=white,dphase dimensions]

\tikzstyle{white rect ddot}=[draw=black,fill=white,doubled,minimum size=5mm,font=\footnotesize,rectangle,rounded corners=2.5mm,inner sep=0.2mm]
\tikzstyle{gray rect ddot}=[draw=black,fill=gray!40!white,doubled,minimum size=6mm,font=\footnotesize,rectangle,rounded corners=3mm]

\tikzstyle{gray phase dot}=[dot,fill=gray!40!white,phase dimensions]
\tikzstyle{gray phase ddot}=[ddot,fill=gray!40!white,dphase dimensions]
\tikzstyle{grey phase dot}=[gray phase dot]
\tikzstyle{grey phase ddot}=[gray phase ddot]

\tikzstyle{small phase dimensions}=[minimum size=4mm,font=\tiny,rectangle,rounded corners=2mm,inner sep=0.2mm,outer sep=-2mm]
\tikzstyle{small dphase dimensions}=[minimum size=4mm,font=\tiny,rectangle,rounded corners=2mm,inner sep=0.2mm,outer sep=-2mm]

\tikzstyle{small gray phase dot}=[dot,fill=gray!40!white,small phase dimensions]
\tikzstyle{small gray phase ddot}=[ddot,fill=gray!40!white,small dphase dimensions]

\tikzstyle{small map}=[draw,shape=rectangle,minimum height=4mm,minimum width=4mm,fill=white]

\tikzstyle{cnot}=[fill=white,shape=circle,inner sep=-1.4pt]

\tikzstyle{asym hadamard}=[fill=white,draw,shape=NEbox,inner sep=0.6mm,font=\footnotesize,minimum height=4mm]
\tikzstyle{asym hadamard conj}=[fill=white,draw,shape=NWbox,inner sep=0.6mm,font=\footnotesize,minimum height=4mm]
\tikzstyle{asym hadamard dag}=[fill=white,draw,shape=SEbox,inner sep=0.6mm,font=\footnotesize,minimum height=4mm]

\tikzstyle{hadamard}=[fill=white,draw,inner sep=0.6mm,font=\footnotesize,minimum height=4mm,minimum width=4mm]
\tikzstyle{small hadamard}=[fill=white,draw,inner sep=0.6mm,minimum height=1.5mm,minimum width=1.5mm]
\tikzstyle{small hadamard rotate}=[small hadamard,rotate=45]
\tikzstyle{dhadamard}=[hadamard,doubled]
\tikzstyle{small dhadamard}=[small hadamard,doubled]
\tikzstyle{small dhadamard rotate}=[small hadamard rotate,doubled]
\tikzstyle{antipode}=[white dot,inner sep=0.3mm,font=\footnotesize]

\tikzstyle{scalar}=[diamond,draw,inner sep=0.5pt,font=\small]
\tikzstyle{dscalar}=[diamond,doubled, draw,inner sep=0.5pt,font=\small]

\tikzstyle{small box}=[rectangle,inline text,fill=white,draw,minimum height=5mm,yshift=-0.5mm,minimum width=5mm,font=\small]
\tikzstyle{small gray box}=[small box,fill=gray!30]
\tikzstyle{medium box}=[rectangle,inline text,fill=white,draw,minimum height=5mm,yshift=-0.5mm,minimum width=10mm,font=\small]
\tikzstyle{square box}=[small box] 
\tikzstyle{medium gray box}=[small box,fill=gray!30]
\tikzstyle{semilarge box}=[rectangle,inline text,fill=white,draw,minimum height=5mm,yshift=-0.5mm,minimum width=12.5mm,font=\small]
\tikzstyle{large box}=[rectangle,inline text,fill=white,draw,minimum height=5mm,yshift=-0.5mm,minimum width=15mm,font=\small]
\tikzstyle{large gray box}=[small box,fill=gray!30]

\tikzstyle{Bayes box}=[rectangle,fill=black,draw, minimum height=3mm, minimum width=3mm]

\tikzstyle{gray square point}=[small box,fill=gray!50]

\tikzstyle{dphase box white}=[dhadamard]
\tikzstyle{dphase box gray}=[dhadamard,fill=gray!50!white]
\tikzstyle{phase box white}=[hadamard]
\tikzstyle{phase box gray}=[hadamard,fill=gray!50!white]

\tikzstyle{point}=[regular polygon,regular polygon sides=3,draw,scale=0.75,inner sep=-0.5pt,minimum width=9mm,fill=white,regular polygon rotate=180]
\tikzstyle{point nosep}=[regular polygon,regular polygon sides=3,draw,scale=0.75,inner sep=-2pt,minimum width=9mm,fill=white,regular polygon rotate=180]
\tikzstyle{copoint}=[regular polygon,regular polygon sides=3,draw,scale=0.75,inner sep=-0.5pt,minimum width=9mm,fill=white]
\tikzstyle{dpoint}=[point,doubled]
\tikzstyle{dcopoint}=[copoint,doubled]

\tikzstyle{pointgrow}=[shape=cornerpoint,kpoint common,scale=0.75,inner sep=3pt]
\tikzstyle{pointgrow dag}=[shape=cornercopoint,kpoint common,scale=0.75,inner sep=3pt]

\tikzstyle{wide copoint}=[fill=white,draw,shape=isosceles triangle,shape border rotate=90,isosceles triangle stretches=true,inner sep=0pt,minimum width=1.5cm,minimum height=6.12mm]
\tikzstyle{wide point}=[fill=white,draw,shape=isosceles triangle,shape border rotate=-90,isosceles triangle stretches=true,inner sep=0pt,minimum width=1.5cm,minimum height=6.12mm,yshift=-0.0mm]
\tikzstyle{wide point plus}=[fill=white,draw,shape=isosceles triangle,shape border rotate=-90,isosceles triangle stretches=true,inner sep=0pt,minimum width=1.74cm,minimum height=7mm,yshift=-0.0mm]

\tikzstyle{wide dpoint}=[fill=white,doubled,draw,shape=isosceles triangle,shape border rotate=-90,isosceles triangle stretches=true,inner sep=0pt,minimum width=1.5cm,minimum height=6.12mm,yshift=-0.0mm]

\tikzstyle{tinypoint}=[regular polygon,regular polygon sides=3,draw,scale=0.55,inner sep=-0.15pt,minimum width=6mm,fill=white,regular polygon rotate=180]

\tikzstyle{white point}=[point]
\tikzstyle{white dpoint}=[dpoint]
\tikzstyle{green point}=[white point] 
\tikzstyle{white copoint}=[copoint]
\tikzstyle{gray point}=[point,fill=gray!40!white]
\tikzstyle{gray dpoint}=[gray point,doubled]
\tikzstyle{red point}=[gray point] 
\tikzstyle{gray copoint}=[copoint,fill=gray!40!white]
\tikzstyle{gray dcopoint}=[gray copoint,doubled]

\tikzstyle{white point guide}=[regular polygon,regular polygon sides=3,font=\scriptsize,draw,scale=0.65,inner sep=-0.5pt,minimum width=9mm,fill=white,regular polygon rotate=180]

\tikzstyle{black point}=[point,fill=black,font=\color{white}]
\tikzstyle{black copoint}=[copoint,fill=black,font=\color{white}]

\tikzstyle{tiny gray point}=[tinypoint,fill=gray!40!white]

\tikzstyle{diredge}=[->]
\tikzstyle{ddiredge}=[<->]
\tikzstyle{rdiredge}=[<-]
\tikzstyle{thickdiredge}=[->, very thick]
\tikzstyle{pointer edge}=[->,very thick,gray]
\tikzstyle{pointer edge part}=[very thick,gray]
\tikzstyle{dashed edge}=[dashed]
\tikzstyle{thick dashed edge}=[very thick,dashed]
\tikzstyle{thick gray dashed edge}=[thick dashed edge,gray!40]
\tikzstyle{thick map edge}=[very thick,|->]

\makeatletter
\newcommand{\boxshape}[3]{%
\pgfdeclareshape{#1}{
\inheritsavedanchors[from=rectangle] 
\inheritanchorborder[from=rectangle]
\inheritanchor[from=rectangle]{center}
\inheritanchor[from=rectangle]{north}
\inheritanchor[from=rectangle]{south}
\inheritanchor[from=rectangle]{west}
\inheritanchor[from=rectangle]{east}
\backgroundpath{
\southwest \pgf@xa=\pgf@x \pgf@ya=\pgf@y
\northeast \pgf@xb=\pgf@x \pgf@yb=\pgf@y

\@tempdima=#2
\@tempdimb=#3

\pgfpathmoveto{\pgfpoint{\pgf@xa - 5pt + \@tempdima}{\pgf@ya}}
\pgfpathlineto{\pgfpoint{\pgf@xa - 5pt - \@tempdima}{\pgf@yb}}
\pgfpathlineto{\pgfpoint{\pgf@xb + 5pt + \@tempdimb}{\pgf@yb}}
\pgfpathlineto{\pgfpoint{\pgf@xb + 5pt - \@tempdimb}{\pgf@ya}}
\pgfpathlineto{\pgfpoint{\pgf@xa - 5pt + \@tempdima}{\pgf@ya}}
\pgfpathclose
}
}}

\boxshape{NEbox}{0pt}{3pt}
\boxshape{SEbox}{0pt}{-3pt}
\boxshape{NWbox}{5pt}{0pt}
\boxshape{SWbox}{-5pt}{0pt}
\boxshape{EBox}{-3pt}{3pt}
\boxshape{WBox}{3pt}{-3pt}
\makeatother

\tikzstyle{cloud}=[shape=cloud,draw,minimum width=1.5cm,minimum height=1.5cm]

\tikzstyle{map}=[draw,shape=NEbox,inner sep=1pt,minimum height=4mm,fill=white]
\tikzstyle{dashedmap}=[draw,dashed,shape=NEbox,inner sep=2pt,minimum height=6mm,fill=white]
\tikzstyle{mapdag}=[draw,shape=SEbox,inner sep=1pt,minimum height=4mm,fill=white]
\tikzstyle{mapadj}=[draw,shape=SEbox,inner sep=2pt,minimum height=6mm,fill=white]
\tikzstyle{maptrans}=[draw,shape=SWbox,inner sep=2pt,minimum height=6mm,fill=white]
\tikzstyle{mapconj}=[draw,shape=NWbox,inner sep=2pt,minimum height=6mm,fill=white]

\tikzstyle{medium map}=[draw,shape=NEbox,inner sep=2pt,minimum height=6mm,fill=white,minimum width=7mm]
\tikzstyle{medium map dag}=[draw,shape=SEbox,inner sep=2pt,minimum height=6mm,fill=white,minimum width=7mm]
\tikzstyle{medium map adj}=[draw,shape=SEbox,inner sep=2pt,minimum height=6mm,fill=white,minimum width=7mm]
\tikzstyle{medium map trans}=[draw,shape=SWbox,inner sep=2pt,minimum height=6mm,fill=white,minimum width=7mm]
\tikzstyle{medium map conj}=[draw,shape=NWbox,inner sep=2pt,minimum height=6mm,fill=white,minimum width=7mm]
\tikzstyle{semilarge map}=[draw,shape=NEbox,inner sep=2pt,minimum height=6mm,fill=white,minimum width=9.5mm]
\tikzstyle{semilarge map trans}=[draw,shape=SWbox,inner sep=2pt,minimum height=6mm,fill=white,minimum width=9.5mm]
\tikzstyle{semilarge map adj}=[draw,shape=SEbox,inner sep=2pt,minimum height=6mm,fill=white,minimum width=9.5mm]
\tikzstyle{semilarge map dag}=[draw,shape=SEbox,inner sep=2pt,minimum height=6mm,fill=white,minimum width=9.5mm]
\tikzstyle{semilarge map conj}=[draw,shape=NWbox,inner sep=2pt,minimum height=6mm,fill=white,minimum width=9.5mm]
\tikzstyle{large map}=[draw,shape=NEbox,inner sep=2pt,minimum height=6mm,fill=white,minimum width=12mm]
\tikzstyle{large map conj}=[draw,shape=NWbox,inner sep=2pt,minimum height=6mm,fill=white,minimum width=12mm]
\tikzstyle{very large map}=[draw,shape=NEbox,inner sep=2pt,minimum height=6mm,fill=white,minimum width=17mm]

\tikzstyle{medium dmap}=[draw,doubled,shape=NEbox,inner sep=2pt,minimum height=6mm,fill=white,minimum width=7mm]
\tikzstyle{medium dmap dag}=[draw,doubled,shape=SEbox,inner sep=2pt,minimum height=6mm,fill=white,minimum width=7mm]
\tikzstyle{medium dmap adj}=[draw,doubled,shape=SEbox,inner sep=2pt,minimum height=6mm,fill=white,minimum width=7mm]
\tikzstyle{medium dmap trans}=[draw,doubled,shape=SWbox,inner sep=2pt,minimum height=6mm,fill=white,minimum width=7mm]
\tikzstyle{medium dmap conj}=[draw,doubled,shape=NWbox,inner sep=2pt,minimum height=6mm,fill=white,minimum width=7mm]
\tikzstyle{semilarge dmap}=[draw,doubled,shape=NEbox,inner sep=2pt,minimum height=6mm,fill=white,minimum width=9.5mm]
\tikzstyle{semilarge dmap trans}=[draw,doubled,shape=SWbox,inner sep=2pt,minimum height=6mm,fill=white,minimum width=9.5mm]
\tikzstyle{semilarge dmap adj}=[draw,doubled,shape=SEbox,inner sep=2pt,minimum height=6mm,fill=white,minimum width=9.5mm]
\tikzstyle{semilarge dmap dag}=[draw,doubled,shape=SEbox,inner sep=2pt,minimum height=6mm,fill=white,minimum width=9.5mm]
\tikzstyle{semilarge dmap conj}=[draw,doubled,shape=NWbox,inner sep=2pt,minimum height=6mm,fill=white,minimum width=9.5mm]
\tikzstyle{large dmap}=[draw,doubled,shape=NEbox,inner sep=2pt,minimum height=6mm,fill=white,minimum width=12mm]
\tikzstyle{large dmap conj}=[draw,doubled,shape=NWbox,inner sep=2pt,minimum height=6mm,fill=white,minimum width=12mm]
\tikzstyle{large dmap trans}=[draw,doubled,shape=SWbox,inner sep=2pt,minimum height=6mm,fill=white,minimum width=12mm]
\tikzstyle{large dmap adj}=[draw,doubled,shape=SEbox,inner sep=2pt,minimum height=6mm,fill=white,minimum width=12mm]
\tikzstyle{large dmap dag}=[draw,doubled,shape=SEbox,inner sep=2pt,minimum height=6mm,fill=white,minimum width=12mm]
\tikzstyle{very large dmap}=[draw,doubled,shape=NEbox,inner sep=2pt,minimum height=6mm,fill=white,minimum width=19.5mm]

\tikzstyle{muxbox}=[draw,shape=rectangle,minimum height=3mm,minimum width=3mm,fill=white]
\tikzstyle{dmuxbox}=[muxbox,doubled]

\tikzstyle{box}=[draw,shape=rectangle,inner sep=2pt,minimum height=6mm,minimum width=6mm,fill=white]
\tikzstyle{dbox}=[draw,doubled,shape=rectangle,inner sep=2pt,minimum height=6mm,minimum width=6mm,fill=white]
\tikzstyle{dmap}=[draw,doubled,shape=NEbox,inner sep=2pt,minimum height=6mm,fill=white]
\tikzstyle{dmapdag}=[draw,doubled,shape=SEbox,inner sep=2pt,minimum height=6mm,fill=white]
\tikzstyle{dmapadj}=[draw,doubled,shape=SEbox,inner sep=2pt,minimum height=6mm,fill=white]
\tikzstyle{dmaptrans}=[draw,doubled,shape=SWbox,inner sep=2pt,minimum height=6mm,fill=white]
\tikzstyle{dmapconj}=[draw,doubled,shape=NWbox,inner sep=2pt,minimum height=6mm,fill=white]

\tikzstyle{ddmap}=[draw,doubled,dashed,shape=NEbox,inner sep=2pt,minimum height=6mm,fill=white]
\tikzstyle{ddmapdag}=[draw,doubled,dashed,shape=SEbox,inner sep=2pt,minimum height=6mm,fill=white]
\tikzstyle{ddmapadj}=[draw,doubled,dashed,shape=SEbox,inner sep=2pt,minimum height=6mm,fill=white]
\tikzstyle{ddmaptrans}=[draw,doubled,dashed,shape=SWbox,inner sep=2pt,minimum height=6mm,fill=white]
\tikzstyle{ddmapconj}=[draw,doubled,dashed,shape=NWbox,inner sep=2pt,minimum height=6mm,fill=white]

\boxshape{sNEbox}{0pt}{3pt}
\boxshape{sSEbox}{0pt}{-3pt}
\boxshape{sNWbox}{3pt}{0pt}
\boxshape{sSWbox}{-3pt}{0pt}
\tikzstyle{smap}=[draw,shape=sNEbox,fill=white]
\tikzstyle{smapdag}=[draw,shape=sSEbox,fill=white]
\tikzstyle{smapadj}=[draw,shape=sSEbox,fill=white]
\tikzstyle{smaptrans}=[draw,shape=sSWbox,fill=white]
\tikzstyle{smapconj}=[draw,shape=sNWbox,fill=white]

\tikzstyle{dsmap}=[draw,dashed,shape=sNEbox,fill=white]
\tikzstyle{dsmapdag}=[draw,dashed,shape=sSEbox,fill=white]
\tikzstyle{dsmaptrans}=[draw,dashed,shape=sSWbox,fill=white]
\tikzstyle{dsmapconj}=[draw,dashed,shape=sNWbox,fill=white]

\boxshape{mNEbox}{0pt}{10pt}
\boxshape{mSEbox}{0pt}{-10pt}
\boxshape{mNWbox}{10pt}{0pt}
\boxshape{mSWbox}{-10pt}{0pt}
\tikzstyle{mmap}=[draw,shape=mNEbox]
\tikzstyle{mmapdag}=[draw,shape=mSEbox]
\tikzstyle{mmaptrans}=[draw,shape=mSWbox]
\tikzstyle{mmapconj}=[draw,shape=mNWbox]

\tikzstyle{mmapgray}=[draw,fill=gray!40!white,shape=mNEbox]
\tikzstyle{smapgray}=[draw,fill=gray!40!white,shape=sNEbox]

\makeatletter

\pgfdeclareshape{cornerpoint}{
\inheritsavedanchors[from=rectangle] 
\inheritanchorborder[from=rectangle]
\inheritanchor[from=rectangle]{center}
\inheritanchor[from=rectangle]{north}
\inheritanchor[from=rectangle]{south}
\inheritanchor[from=rectangle]{west}
\inheritanchor[from=rectangle]{east}
\backgroundpath{
\southwest \pgf@xa=\pgf@x \pgf@ya=\pgf@y
\northeast \pgf@xb=\pgf@x \pgf@yb=\pgf@y

\pgfmathsetmacro{\pgf@shorten@left}{\pgfkeysvalueof{/tikz/shorten left}}
\pgfmathsetmacro{\pgf@shorten@right}{\pgfkeysvalueof{/tikz/shorten right}}

\pgfpathmoveto{\pgfpoint{0.5 * (\pgf@xa + \pgf@xb)}{\pgf@ya - 5pt}}
\pgfpathlineto{\pgfpoint{\pgf@xa - 8pt + \pgf@shorten@left}{\pgf@yb - 1.5 * \pgf@shorten@left}}
\pgfpathlineto{\pgfpoint{\pgf@xa - 8pt + \pgf@shorten@left}{\pgf@yb}}
\pgfpathlineto{\pgfpoint{\pgf@xb + 8pt - \pgf@shorten@right}{\pgf@yb}}
\pgfpathlineto{\pgfpoint{\pgf@xb + 8pt - \pgf@shorten@right}{\pgf@yb - 1.5 * \pgf@shorten@right}}
\pgfpathclose
}
}

\pgfdeclareshape{cornercopoint}{
\inheritsavedanchors[from=rectangle] 
\inheritanchorborder[from=rectangle]
\inheritanchor[from=rectangle]{center}
\inheritanchor[from=rectangle]{north}
\inheritanchor[from=rectangle]{south}
\inheritanchor[from=rectangle]{west}
\inheritanchor[from=rectangle]{east}
\backgroundpath{
\southwest \pgf@xa=\pgf@x \pgf@ya=\pgf@y
\northeast \pgf@xb=\pgf@x \pgf@yb=\pgf@y

\pgfmathsetmacro{\pgf@shorten@left}{\pgfkeysvalueof{/tikz/shorten left}}
\pgfmathsetmacro{\pgf@shorten@right}{\pgfkeysvalueof{/tikz/shorten right}}

\pgfpathmoveto{\pgfpoint{0.5 * (\pgf@xa + \pgf@xb)}{\pgf@yb + 5pt}}
\pgfpathlineto{\pgfpoint{\pgf@xa - 8pt + \pgf@shorten@left}{\pgf@ya + 1.5 * \pgf@shorten@left}}
\pgfpathlineto{\pgfpoint{\pgf@xa - 8pt + \pgf@shorten@left}{\pgf@ya}}
\pgfpathlineto{\pgfpoint{\pgf@xb + 8pt - \pgf@shorten@right}{\pgf@ya}}
\pgfpathlineto{\pgfpoint{\pgf@xb + 8pt - \pgf@shorten@right}{\pgf@ya + 1.5 * \pgf@shorten@right}}
\pgfpathclose
}
}

\makeatother

\pgfkeyssetvalue{/tikz/shorten left}{0pt}
\pgfkeyssetvalue{/tikz/shorten right}{0pt}

\tikzstyle{kpoint common}=[draw,fill=white,inner sep=1pt,minimum height=4mm]
\tikzstyle{kpoint sc}=[shape=cornerpoint,kpoint common]
\tikzstyle{kpoint adjoint sc}=[shape=cornercopoint,kpoint common]
\tikzstyle{kpoint}=[shape=cornerpoint,shorten left=5pt,kpoint common]
\tikzstyle{kpoint adjoint}=[shape=cornercopoint,shorten left=5pt,kpoint common]
\tikzstyle{kpoint conjugate}=[shape=cornerpoint,shorten right=5pt,kpoint common]
\tikzstyle{kpoint transpose}=[shape=cornercopoint,shorten right=5pt,kpoint common]
\tikzstyle{kpoint symm}=[shape=cornerpoint,shorten left=5pt,shorten right=5pt,kpoint common]

\tikzstyle{wide kpoint sc}=[shape=cornerpoint,kpoint common, minimum width=1 cm]
\tikzstyle{wide kpointdag sc}=[shape=cornercopoint,kpoint common, minimum width=1 cm]

\tikzstyle{black kpoint}=[shape=cornerpoint,shorten left=5pt,kpoint common,fill=black,font=\color{white}]

\tikzstyle{black kpoint sm}=[shape=cornerpoint,shorten left=5pt,kpoint common,fill=black,font=\color{white},scale=0.75]

\tikzstyle{black kpoint adjoint}=[shape=cornercopoint,shorten left=5pt,kpoint common,fill=black,font=\color{white}]
\tikzstyle{black kpointadj}=[shape=cornercopoint,shorten left=5pt,kpoint common,fill=black,font=\color{white}]

\tikzstyle{black kpointadj sm}=[shape=cornercopoint,shorten left=5pt,kpoint common,fill=black,font=\color{white},scale=0.75]

\tikzstyle{black dkpoint}=[shape=cornerpoint,shorten left=5pt,kpoint common,fill=black, doubled,font=\color{white}]
\tikzstyle{black dkpoint adjoint}=[shape=cornercopoint,shorten left=5pt,kpoint common,fill=black, doubled,font=\color{white}]
\tikzstyle{black dkpointadj}=[shape=cornercopoint,shorten left=5pt,kpoint common,fill=black, doubled,font=\color{white}]

\tikzstyle{black dkpoint sm}=[shape=cornerpoint,shorten left=5pt,kpoint common,fill=black, doubled,font=\color{white},scale=0.75]
\tikzstyle{black dkpointadj sm}=[shape=cornercopoint,shorten left=5pt,kpoint common,fill=black, doubled,font=\color{white},scale=0.75]

\tikzstyle{kpointdag}=[kpoint adjoint]
\tikzstyle{kpointadj}=[kpoint adjoint]
\tikzstyle{kpointconj}=[kpoint conjugate]
\tikzstyle{kpointtrans}=[kpoint transpose]

\tikzstyle{big kpoint}=[kpoint, minimum width=1.2 cm, minimum height=8mm, inner sep=4pt, text depth=3mm]

\tikzstyle{wide kpoint}=[kpoint, minimum width=1 cm, inner sep=2pt]
\tikzstyle{wide kpointdag}=[kpointdag, minimum width=1 cm, inner sep=2pt]
\tikzstyle{wide kpointconj}=[kpointconj, minimum width=1 cm, inner sep=2pt]
\tikzstyle{wide kpointtrans}=[kpointtrans, minimum width=1 cm, inner sep=2pt]

\tikzstyle{wider kpoint}=[kpoint, minimum width=1.25 cm, inner sep=2pt]
\tikzstyle{wider kpointdag}=[kpointdag, minimum width=1.25 cm, inner sep=2pt]
\tikzstyle{wider kpointconj}=[kpointconj, minimum width=1.25 cm, inner sep=2pt]
\tikzstyle{wider kpointtrans}=[kpointtrans, minimum width=1.25 cm, inner sep=2pt]

\tikzstyle{gray kpoint}=[kpoint,fill=gray!50!white]
\tikzstyle{gray kpointdag}=[kpointdag,fill=gray!50!white]
\tikzstyle{gray kpointadj}=[kpointadj,fill=gray!50!white]
\tikzstyle{gray kpointconj}=[kpointconj,fill=gray!50!white]
\tikzstyle{gray kpointtrans}=[kpointtrans,fill=gray!50!white]

\tikzstyle{gray dkpoint}=[kpoint,fill=gray!50!white,doubled]
\tikzstyle{gray dkpointdag}=[kpointdag,fill=gray!50!white,doubled]
\tikzstyle{gray dkpointadj}=[kpointadj,fill=gray!50!white,doubled]
\tikzstyle{gray dkpointconj}=[kpointconj,fill=gray!50!white,doubled]
\tikzstyle{gray dkpointtrans}=[kpointtrans,fill=gray!50!white,doubled]

\tikzstyle{white label}=[draw,fill=white,rectangle,inner sep=0.7 mm]
\tikzstyle{gray label}=[draw,fill=gray!50!white,rectangle,inner sep=0.7 mm]
\tikzstyle{black label}=[draw,fill=black,rectangle,inner sep=0.7 mm]

\tikzstyle{dkpoint}=[kpoint,doubled]
\tikzstyle{wide dkpoint}=[wide kpoint,doubled]
\tikzstyle{dkpointdag}=[kpoint adjoint,doubled]
\tikzstyle{wide dkpointdag}=[wide kpointdag,doubled]
\tikzstyle{dkcopoint}=[kpoint adjoint,doubled]
\tikzstyle{dkpointadj}=[kpoint adjoint,doubled]
\tikzstyle{dkpointconj}=[kpoint conjugate,doubled]
\tikzstyle{dkpointtrans}=[kpoint transpose,doubled]

\tikzstyle{kscalar}=[kpoint common, shape=EBox, inner xsep=-1pt, inner ysep=3pt,font=\small]
\tikzstyle{kscalarconj}=[kpoint common, shape=WBox, inner xsep=-1pt, inner ysep=3pt,font=\small]

\tikzstyle{spekpoint}=[kpoint sc,minimum height=5mm,inner sep=3pt]
\tikzstyle{spekcopoint}=[kpoint adjoint sc,minimum height=5mm,inner sep=3pt]

\tikzstyle{dspekpoint}=[spekpoint,doubled]
\tikzstyle{dspekcopoint}=[spekcopoint,doubled]

 \tikzstyle{upground}=[circuit ee IEC,thick,ground,rotate=90,scale=2.5]
 \tikzstyle{downground}=[circuit ee IEC,thick,ground,rotate=-90,scale=2.5]
 \tikzstyle{bigground}=[regular polygon,regular polygon sides=3,draw=gray,scale=0.50,inner sep=-0.5pt,minimum width=10mm,fill=gray]

\tikzstyle{arrs}=[-latex,font=\small,auto]
\tikzstyle{arrow plain}=[arrs]
\tikzstyle{arrow dashed}=[dashed,arrs]
\tikzstyle{arrow bold}=[very thick,arrs]
\tikzstyle{arrow hide}=[draw=white!0,-]
\tikzstyle{arrow reverse}=[latex-]
\tikzstyle{cdnode}=[]

\DeclareMathOperator{\tr}{tr}

\let\olddagger\dagger
\renewcommand{\dagger}{\ensuremath{\olddagger}\xspace}

\theoremstyle{definition}
\newtheorem{theorem}{Theorem}[section]
\newtheorem{corollary}[theorem]{Corollary}

\newtheorem{example}[]{Example}

\newtheorem{example*}[theorem]{Example*}
\newtheorem{examples*}[theorem]{Examples*}

\newtheorem{remark*}[theorem]{Remark*}

\theoremstyle{plain}

\usepackage{color}
\def\bR{\begin{color}{red}}
\def\bB{\begin{color}{blue}}
\def\bM{\begin{color}{magenta}}
\def\bC{\begin{color}{cyan}}
\def\bW{\begin{color}{white}}
\def\bBl{\begin{color}{black}}
\def\bG{\begin{color}{green}}
\def\bY{\begin{color}{yellow}}
\def\e{\end{color}\xspace}
\newcommand{\bit}{\begin{itemize}}
\newcommand{\eit}{\end{itemize}\par\noindent}
\newcommand{\ben}{\begin{enumerate}}
\newcommand{\een}{\end{enumerate}\par\noindent}
\newcommand{\beq}{\begin{equation}}
\newcommand{\eeq}{\end{equation}\par\noindent}
\newcommand{\beqa}{\begin{eqnarray*}}
\newcommand{\eeqa}{\end{eqnarray*}\par\noindent}
\newcommand{\beqn}{\begin{eqnarray}}
\newcommand{\eeqn}{\end{eqnarray}\par\noindent}

\def\jR{\begin{color}{DarkRed}}
\def\jB{\begin{color}{MidnightBlue}}
\def\jM{\begin{color}{magenta}}
\def\jC{\begin{color}{cyan}}
\def\jW{\begin{color}{white}}
\def\jBl{\begin{color}{black}}
\def\jG{\begin{color}{green}}
\def\jY{\begin{color}{yellow}}

\def\cR{\begin{color}{Crimson}}
\def\cB{\begin{color}{cyan}}
\def\cM{\begin{color}{magenta}}
\def\cC{\begin{color}{cyan}}
\def\cW{\begin{color}{white}}
\def\cBl{\begin{color}{black}}
\def\cG{\begin{color}{green}}
\def\cY{\begin{color}{yellow}}

\usepackage{natbib}
\usepackage{wasysym}
\usepackage{soul}
\usepackage{bm}
\usepackage{upgreek}
\usepackage[linktocpage=true, colorlinks=true, linkcolor=blue, urlcolor=blue,
	citecolor=blue]{hyperref}
\usepackage{paralist}
\usepackage{soul}
\usepackage{changes}
\usepackage{xr}

\newtheorem{defn}[theorem]{Definition}

\newtheorem{lem}[theorem]{Lemma}
\newtheorem{rem}[theorem]{Remark}
\newtheorem{prop}[theorem]{Proposition}

\definecolor{compositionalitypink}{RGB}{255,90,95}

\definecolor{compositionalitydarkpink}{RGB}{255,70,76}
\definecolor{quantumgray}{RGB}{85, 85, 85}
\definecolor{quantumviolet}{RGB}{83, 37, 127} 

\allowdisplaybreaks

\begin{document}

\title{Every non-signalling channel is common-cause realizable}

\author{Paulo J. Cavalcanti}
\email{paulojcvf@gmail.com}
\affiliation{International Centre for Theory of Quantum Technologies,
	University of Gda\'nsk, 80-309 Gda\'nsk, Poland}
\author{John H.~Selby}
\email{john.h.selby@gmail.com}
\affiliation{International Centre for Theory of Quantum Technologies,
	University of Gda\'nsk, 80-309 Gda\'nsk, Poland}
\author{Ana Bel\'en Sainz}
\email{ana.sainz@ug.edu.pl}
\affiliation{International Centre for Theory of Quantum Technologies,
	University of Gda\'nsk, 80-309 Gda\'nsk, Poland}

\date{\today}

\begin{abstract}
	In this work we show that the set of non-signalling resources of a
	locally-tomographic  generalised probabilistic theory (GPT), such as
	quantum
	and classical theory, coincides with its set of GPT-common-cause
	realizable
	resources, where the common causes come from an associated GPT. From a
	causal
	perspective, this result provides a reason for, in the study of
	resource
	theories of common-cause processes, taking the non-signalling channels
	as the
	resources of the enveloping theory. This answers a critical open
	question in
	Ref.~\cite{schmid2020postquantum}. An immediate corollary of our result
	is that
	every non-signalling assemblage is realizable in a GPT, answering in
	the
	affirmative the question posed in Ref.~\cite{cavalcanti2022post}.
\end{abstract}

\maketitle
\tableofcontents

\section{Introduction}
There has been a great deal of recent interest in the study of resource
theories \cite{coecke2016mathematical} in which the free operations are either
Local Operations and Shared Randomness (LOSR)
\cite{de2014nonlocality,GellerPiani,cowpie,schmid2020standard,EPRLOSR,
	zjawin2022resource}, for the purposes of studying nonlocality and
entanglement,
or Local Operations and Shared Entanglement (LOSE)
\cite{schmid2020postquantum}, for the purposes of studying post-quantum
nonlocality.
In particular, it has been shown that these can be studied in a
type-independent manner
\cite{schmid2019typeindependent,rosset2019characterizing,schmid2020postquantum}
such that resources of various types (entangled states, nonlocal boxes,
steerable assemblages, etc.) can be treated in a uniform and unified way.
These resource theories are motivated by the idea that the best way to
understand Bell's theorem is from the perspective of causal models
\cite{Wood_2015,cowpie}, and that the lesson to be learnt from Bell's theorem
is that we need an intrinsically quantum notion of causality and of common
causes \cite{schmid2020unscrambling,cowpie}.

Defining a resource theory requires a specification of both a \emph{free} and
an \emph{enveloping} theory \cite{coecke2016mathematical}. The free theory
specifies the things that can be done effectively without cost, whilst the
enveloping theory specifies the things that can be done irrespective of cost.
Whilst in the study of LOSE and LOSR it is clear how the free theory should be
defined, it is not clear how the \emph{enveloping theory} should be defined
\cite{schmid2020postquantum}.
There are two options for this, each of which has pros and cons. On the one
hand, we have the choice which is typically made, which is to use the
enveloping theory which describes \emph{non-signalling} resources. The benefit
of this choice is that it is mathematically simple to characterise, since in
the cases of interest so far the set of such resources can often be expressed
in a computationally-easy way (polytope, or semidefinite programme)
\cite{BellRev,sainz2015postquantum}. Its downside, however, is that this
enveloping theory is not so well motivated from a causal perspective -- it
makes sense to say that resources should be non-signalling, but why should
\emph{all} non-signalling resources be considered? On the other hand, we can
take the enveloping theory to describe	arbitrary \emph{common cause}
resources, typically described using the framework of generalised probabilistic
theories (GPTs) subsuming classical and quantum common causes as special cases.
The benefit of this approach is that it is conceptually well motivated, from
the causal perspective \cite{cowpie}. Its downside, however, is that providing
a clean mathematical characterisation of this enveloping theory is an open
problem. The characterisation and the relationship between these two options
was cleanly articulated as an open question in Ref.~\cite[Open Question
	1]{schmid2020postquantum}.

In this paper we resolve the tension between these two choices, by showing that
these two options actually coincide. This means that we get the benefits of
both approaches with none of the downsides. It is well established that every
common-cause realisable resource is non-signalling, so here we just focus on
the converse direction. In particular, we  show  that there exists a GPT in
which all non-signalling resources of a target locally tomographic GPT, such as
quantum theory, can be realised in a common-cause setting.
On the one hand, we can view this result as, for the first time, providing a
clear characterisation of the set of GPT-realisable resources. On the other
hand, we can also view it as providing a principled justification, backed by
the causal perspective, for choosing the set of non-signalling resources as the
enveloping theory in resource theories of common-cause processes.
We moreover show that this result holds not only in the bipartite case, which
has so far dominated the literature, but also in the general multipartite
scenario, thereby setting the stage for explorations of multipartite
generalisations of LOSR and LOSE resource theories.
A corollary of this result answers one of the open questions posed in
Ref.~\cite{cavalcanti2022post}, namely it shows that indeed any non-signalling
assemblage can be given a GPT-common-cause explanation.

The scheme by which we build the GPT where all non-signalling resources can be
realised in a common-cause setting differs from the standard approach to GPT
construction in the literature.
Usually, GPTs are constructed by making reference to the geometry of their
states, effects, and transformations spaces, requiring, for example, that they
are convex subsets of linear spaces (see, e.g., Ref.~\cite{barrett2007gpt}).
Here, instead of putting emphasis on the geometry, we focus our attention on
compositionality, that is we take a process-theoretic
\cite{coecke2011universe,coecke2018picturing,gogioso2017categorical,selby2018reconstructing}
approach to constructing one GPT from another.
By focusing on the compositional properties of the theory, our method also has
the potential to be applied to other problems.

To be more formal, let us define a \emph{common-cause completion} of a given
GPT $\mathbf{G}$ as an supertheory of $\mathbf{G}$ which can realise all of the
non-signalling resources $\mathbf{G}$ in a common-cause scenario. If some
theory is the common-cause completion of itself, then we call it
\emph{common-cause complete}, in contrast to quantum and classical theory which
have non-signalling resources which cannot be realised in common cause
scenarios, being  therefore common-cause incomplete. In this paper, we define a
common-cause completion map, $\mathcal{C}$, which takes an arbitrary
tomographically-local GPT $\mathbf{G}$ as an input, and gives a common-cause
completion of it, $\mathcal{C}[\mathbf{G}]$, as output. Specifically, this
means that $\mathbf{G}$ is a full subtheory of $\mathcal{C}[\mathbf{G}]$ and
that every non-signalling resource in $\mathbf{G}$ can be realised with only
common-cause resources in $\mathcal{C}[\mathbf{G}]$. Proving the existence of
such a common-cause completion map demonstrates the main claims of this paper:
i) all non-signalling resources in $\mathbf{G}$ are common-cause realisable in
the GPT $\mathcal{C}[\mathbf{G}]$, and so the non-signalling resources in
$\mathbf{G}$ coincide with its GPT-common-cause realisable processes; and ii)
every non-signalling assemblage in $\mathbf{G}$ is realisable in an EPR
scenario in $\mathcal{C}[\mathbf{G}]$.

\section{Generalised probabilistic theories (GPTs)}\label{GPTs}

In this section, we provide a concise overview of Generalized Probabilistic
Theories (GPTs) \cite{hardy2001quantum,barrett2007gpt}, emphasizing their
compositional attributes. We provide a brief introduction here, and refer the
interested reader to, for example,
Refs.~\cite{muller2021probabilistic,plavala2021general} for more details.
Specifically, we are following the formalism of
Refs.~\cite{gogioso2017categorical,selby2018reconstructing}.

Conceptually, a GPT is a theory about experiments that assigns probabilities to
observation events, equipped with a compositional structure that mirrors the
possibility we have to perform actions sequentially or in parallel. Formally
speaking, the compositional aspects of the theory are captured by the fact that
a GPT is a (strict) symmetric monoidal category (SMC) (see App.~\ref{ap:SMC}).
The probabilistic aspects are captured by the fact that we have a classical
(stochastic) interface with the full theory in order to represent outcomes and
control variables, formally, this means that we have the SMC $\mathbf{Stoch}$
(Sec.~\ref{sec:stoch}) as a full subtheory. This leads to a convex structure
(Sec.~\ref{sec:conv}) on the sets of processes with a given input and output,
and allows us to define a notion of tomography (Sec.~\ref{sec:tomog}). Finally,
we capture the requirement that the theory interact well with relativistic
causal structure, by demanding the existence of unique discarding maps
(Sec.~\ref{sec:discard}).

In the rest of this section, we will introduce the diagrammatic notation used
throughout this work, and discuss the defining features of a GPT that we
mentioned above.

\subsection{Diagrammatic notation}\label{diagrams}
An interesting feature of SMCs is that they have a diagrammatic representation
with which we can perform every calculation that we could using their axiomatic
definition \cite{joyal1991geometry,selinger2010survey,patterson2021wiring}.
In the context of GPTs, we can represent their processes as boxes with input
and output wires, and encode the composition of these processes by how they are
wired together.

In the diagrammatic notation, each wire is named to represent a system type,
and we follow the convention where those connected to the bottom of the boxes
represent the input types of the process, while those at the top are the
outputs. Note that this means that, in our convention, ``time'' in the diagrams
flows from the bottom up.
In this way we can represent a process $f: A \to  B$, that takes a system of
type $A$ to a system of type $B$, as follows:
\begin{equation}\label{}
	\quad f:A \to  B\quad \doteq\quad %
\InputIfFileExists{Diagrams/f.tikz}{}{\input{./figures/Diagrams/f.tikz}},
\end{equation}
where we are using $\doteq$ to indicate the translation from one notation into
another.

We often omit wire labels for simplicity, and/or use colors to encode certain
information about the system type. For instance, in this paper we will use
\begin{equation}\label{}
\begin{tikzpicture}
	\begin{pgfonlayer}{nodelayer}
		\node [style=none] (0) at (0, 0.75) {};
		\node [style=none] (1) at (0, -0.75) {};
		\node [style=none] (4) at (0, 1.25) {};
		\node [style=none] (5) at (0, -1.25) {};
	\end{pgfonlayer}
	\begin{pgfonlayer}{edgelayer}
		\draw [cWire] (0.center) to (1.center);
	\end{pgfonlayer}
\end{tikzpicture}
},\quad %
\begin{tikzpicture}
	\begin{pgfonlayer}{nodelayer}
		\node [style=none] (0) at (0, 0.75) {};
		\node [style=none] (1) at (0, -0.75) {};
		\node [style=none] (2) at (0, 1.25) {};
		\node [style=none] (3) at (0, -1.25) {};
	\end{pgfonlayer}
	\begin{pgfonlayer}{edgelayer}
		\draw [style=qWire] (0.center) to (1.center);
	\end{pgfonlayer}
\end{tikzpicture}
},\quad
\begin{tikzpicture}
	\begin{pgfonlayer}{nodelayer}
		\node [style=none] (0) at (0, 0.75) {};
		\node [style=none] (1) at (0, -0.75) {};
		\node [style=none] (2) at (0, 1.25) {};
		\node [style=none] (3) at (0, -1.25) {};
	\end{pgfonlayer}
	\begin{pgfonlayer}{edgelayer}
		\draw [style=pWire, in=90, out=-90] (0.center) to (1.center);
	\end{pgfonlayer}
\end{tikzpicture}
},\quad \text{ or } \quad
\begin{tikzpicture}
	\begin{pgfonlayer}{nodelayer}
		\node [style=none] (0) at (0, 0.75) {};
		\node [style=none] (1) at (0, -0.75) {};
		\node [style=none] (2) at (0, 1.25) {};
		\node [style=none] (3) at (0, -1.25) {};
	\end{pgfonlayer}
	\begin{pgfonlayer}{edgelayer}
		\draw [style=gWire] (0.center) to (1.center);
	\end{pgfonlayer}
\end{tikzpicture}
} \,,
\end{equation}
where, for example, the first of these represents a classical system of
unspecified dimension, and the meaning of the others will be explained in
section \ref{results}.
To represent composite types such as $A \otimes B$, we just put their wires
side by side, as in
\begin{equation}\label{}
	A \otimes B\quad \doteq\quad %
\InputIfFileExists{Diagrams/AB.tikz}{}{\input{./figures/Diagrams/AB.tikz}} \,.
\end{equation}
Using this notation for composite systems, a process with composite input or
output wires is depicted as having multiple input/output wires, e.g.,
\begin{equation}\label{}
	f: A \otimes B \to C \otimes D \otimes E\quad \doteq\quad
\InputIfFileExists{Diagrams/fabcde.tikz}{}{\input{./figures/Diagrams/fabcde.tikz}} .
\end{equation}

One system type that every GPT must contain is the trivial system, which
corresponds to having no system at all. We refer to it in text as $I$. Since
the trivial system is the unit for parallel composition (i.e., the monoidal
unit of the symmetric monoidal category), we have $A \otimes I = A =  I \otimes
	A$, diagrammatically, $I$ is represented by empty space:
\begin{equation}\label{}
	I \quad \doteq\quad  %
\begin{tikzpicture}
	\begin{pgfonlayer}{nodelayer}
		\node [style=none] (0) at (0, 0.75) {};
		\node [style=none] (1) at (0, -0.75) {};
	\end{pgfonlayer}
	\begin{pgfonlayer}{edgelayer}
		\draw [style=thick gray dashed edge] (0.center) to (1.center);
	\end{pgfonlayer}
\end{tikzpicture}
} \,.
\end{equation}
States and effects can be seen as preparation and observation procedures,
respectively, which are processes that start and end in the trivial system,
i.e., they must not have input or output wires respectively. For example, if
$s$ is a state and $e$ an effect, then we denote them as
\begin{equation}\label{}
	s:I\to A\ \ \doteq\ \ %
\begin{tikzpicture}
	\begin{pgfonlayer}{nodelayer}
		\node [style=point] (0) at (0, 0) {$s$};
		\node [style=none] (5) at (0, 1.25) {};
	\end{pgfonlayer}
	\begin{pgfonlayer}{edgelayer}
		\draw [qWire] (5.center) to (0);
	\end{pgfonlayer}
\end{tikzpicture}
} \quad\text{ and} \quad
	e:A\to I\ \ \doteq\ \  %
\begin{tikzpicture}
	\begin{pgfonlayer}{nodelayer}
		\node [style=copoint] (4) at (0, 0) {$e$};
		\node [style=none] (5) at (0, -1.25) {};
	\end{pgfonlayer}
	\begin{pgfonlayer}{edgelayer}
		\draw [qWire] (5.center) to (4);
	\end{pgfonlayer}
\end{tikzpicture}
}\,.
\end{equation}
There can also be processes with both input and output as the trivial system,
$p: I \to I$, which are represented by diagrams without open wires.
The compositional properties of the SMC imply that diagrams of this kind can be
composed together with a multiplicative structure, and hence can be called
numbers. For instance, we could have
\begin{equation}\label{probabilityDiagram}
	1: I \to I\quad \doteq\quad %
\begin{tikzpicture}
	\begin{pgfonlayer}{nodelayer}
		\node [style=scalar] (0) at (0, 0) {$1$};
	\end{pgfonlayer}
	\begin{pgfonlayer}{edgelayer}
	\end{pgfonlayer}
\end{tikzpicture}
}.
\end{equation}

Finally, we represent the parallel composition $ \otimes $ of processes by
drawing their boxes side by side, and their sequential composition $ \circ$ by
connecting the input and output wires of matching types. That is, for $f:A \to
	B$ and $g:C \to D$
\begin{equation}\label{}
	f \otimes g:A \otimes C \to B \otimes D\quad \doteq\quad
\InputIfFileExists{Diagrams/ftg.tikz}{}{\input{./figures/Diagrams/ftg.tikz}} \,,
\end{equation}
and  for $f:A \to B$ and $g:B \to C$
\begin{equation}\label{}
	g \circ f : A \to C\quad \doteq\quad %
\InputIfFileExists{Diagrams/gof.tikz}{}{\input{./figures/Diagrams/gof.tikz}}\,.
\end{equation}
One example of a more complex diagram is
\begin{equation}\label{}
\InputIfFileExists{Diagrams/abcde.tikz}{}{\input{./figures/Diagrams/abcde.tikz}}\,,
\end{equation}
where we omit the labels of the wires, but it should be understood that
connections are allowed only when types match.

This notation and the rules for composing diagrams are common to all (strict)
symmetric monoidal categories. Now, it remains to discuss features that are
shared only by those who can be considered as GPTs.
Since one of the ingredients of a GPT is that they contain $\mathbf{Stoch}$ as
a full subtheory, we start from the definition of that theory.

\subsection{Example: classical stochastic maps}\label{sec:stoch}

As we mentioned, any GPT must have $\mathbf{Stoch}$ as a full subtheory. The
simplest possible GPT, then is the one that contains nothing else (if the other
properties are satisfied, of course, which is the case).

In order to define $\mathbf{Stoch}$, all we have to do is to define what
concrete mathematical objects correspond to its system types, states, effects,
transformations, and composition rules (parallel and sequential composition).
We organized this information in the following table:

\begin{center}
	\begin{tabular}{||c|c|c||}
		\hline
		Element                & Definition
		                       & Example                               \\
		\hline
		System types           & Real vector spaces
		                       & $ \mathds{R}^{2}$                     \\
		\hline
		States                 & Probability column vectors
		                       & $\begin{pmatrix} 1/2 \\ 1/2

			                          \end{pmatrix}$            \\
		\hline
		Effects                & Row vectors whose all entries are
		equal to 1             & $\begin{pmatrix} 1

				                           & 1\end{pmatrix}
		$
		\\
		\hline
		Transformations        & Stochastic matrices
		                       & $\begin{pmatrix} 1/2 & 1/3 \\ 1/2 &
                   2/3\end{pmatrix}
		$
		\\
		\hline
		Sequential Composition & Matrix multiplication
		                       & $\begin{pmatrix} 1/2 & 1/3
                \\ 1/2 & 2/3\end{pmatrix}
			\begin{pmatrix} 1/2 \\ 1/2 \end{pmatrix}
		$
		\\
		\hline
		Parallel Composition   & Kronecker product (or tensor product)
		                       &
		$\begin{pmatrix} 1/2 \\ 2/3 \end{pmatrix} \otimes
		\begin{pmatrix} 1/2 \\ 1/2

			\end{pmatrix} $                                      \\
		\hline
	\end{tabular}
\end{center}

Because we require that GPTs have this theory as a full subtheory, it will act
as an interface to provide the GPT with the probabilistic interpretation that
we need. For example, in this framework we describe a measurement as a process
from a general system to a system in $\mathbf{Stoch}$.

\subsection{Defining properties of causal GPTs}
Not every SMC can be considered as being a hypothetical theory of physics. In
this section, we characterise those that can. In particular, what we are
looking for with this characterization is to use $\mathbf{Stoch}$ as an
interface to the theory that enables us to make statistical predictions in a
manner coherent with its compositional structure, and where we can characterize
the objects by the statistics that they can generate.

The additional features that an SMC has to satisfy in order to be a causal GPT
are:
\begin{enumerate}
	\item The SMC contains \textbf{$\mathbf{Stoch}$ as a full subtheory}.
	\item There is a \textbf{convex structure} compatible with the one from
	      $\mathbf{Stoch}$.
	\item There is a notion of \textbf{tomography}.
	\item There is a \textbf{unique effect} associated to each system type.
\end{enumerate}

In this manuscript we further focus on	GPTs that satisfy the following
additional property:
\begin{enumerate}
	\item[5.] The theory is \textbf{locally tomographic}.
\end{enumerate}

We now discuss each of those points in turn.

\subsubsection{\textbf{$\mathbf{Stoch}$ is a full subtheory}}

This means that all of the systems from $\mathbf{Stoch}$ and all of the
processes from $\mathbf{Stoch}$ are also in the GPT, and, moreover, that when
we compose these systems and processes in the GPT this matches the composition
in $\mathbf{Stoch}$ \cite{gogioso2017categorical}. Moreover, if we have a
process in the GPT which only has inputs and outputs coming from
$\mathbf{Stoch}$, then this must  be a process coming from $\mathbf{Stoch}$.

The importance of that, is that inside a GPT, we can take the maps that go from
a classical system (i.e. a system interpreted as a system of $\mathbf{Stoch}$
to another one as a stochastic process. Then, these processes, with all their
internal probabilities, provide a probabilistic interpretation to the diagrams.
Note that if it were not a \emph{full} subtheory, then there would necessarily
be situations in which the theory failed to make sensible probabilistic
predictions, for example, giving negative probabilities for measurement
outcomes.

For example, suppose we have a state of some general system in the GPT, $A$,
then a (destructive) measurement for $A$ would be a process with $A$ as an
input and some system $X$ in $\mathbf{Stoch}$ as an output, when we compose
these we are left with a process which must be a state in $\mathbf{Stoch}$,
namely, a probability distribution. It is precisely these probability
distributions which encode the probabilistic predictions of the GPT.

We denote the systems coming from the subtheory $\mathbf{Stoch}$ as:
\begin{equation}
\begin{tikzpicture}
	\begin{pgfonlayer}{nodelayer}
		\node [style=none] (0) at (0, 0.75) {};
		\node [style=none] (1) at (0, -0.75) {};
		\node [style=right label] (2) at (0, -0.5) {$X$};
	\end{pgfonlayer}
	\begin{pgfonlayer}{edgelayer}
		\draw [cWire] (0.center) to (1.center);
	\end{pgfonlayer}
\end{tikzpicture}
}\ ,
\end{equation}
where we use a thin gray wire to distinguish the systems in the subtheory from
generic systems in the GPT.

Note that in many other approaches to GPTs, the probabilities are encoded as
scalars in the theory. In the approach we take here this is not the case, as,
in particular, we find here that there is a unique scalar, the number $1$.
Instead, we obtain probability distributions over measurement outcomes via the
states of the subtheory, $\mathbf{Stoch}$. For example, this is what we obtain
when we compose a state of a generic system in sequence with a measurement on
that system.

\subsubsection{\textbf{Convex structure}} \label{sec:conv}
In order to naturally express statistical mixtures in the GPTs, we require them
to be closed under convex mixtures of processes of matching input and output
types. We require further that this composition is consistent with the convex
composition from $\mathbf{Stoch}$ \cite{gogioso2017categorical}. To start
illustrating that, note that if we have $f:A \to B$ and $g: A \to B$, there
must exist some $pf+(1-p)g: A \to B$ in the theory where we denote this as
\begin{equation}\label{}
	p\ %
\InputIfFileExists{Diagrams/f.tikz}{}{\input{./figures/Diagrams/f.tikz}} + (1-p)\ %
\InputIfFileExists{Diagrams/g.tikz}{}{\input{./figures/Diagrams/g.tikz}} \ \ =\ \
\InputIfFileExists{Diagrams/fpg.tikz}{}{\input{./figures/Diagrams/fpg.tikz}}\,.
\end{equation}
Note that these combinations are allowed only when the input/output systems are
the same for each of the combined processes. Moreover, these must distribute
over diagrams, that is, they must satisfy, for example:
\begin{equation}\label{}
\InputIfFileExists{Diagrams/sum1.tikz}{}{\input{./figures/Diagrams/sum1.tikz}}\ \ =\ \ %
\InputIfFileExists{Diagrams/sum2.tikz}{}{\input{./figures/Diagrams/sum2.tikz}} \,.
\end{equation}
Finally, these convex combinations must match up with the standard notion of
convex-combinations when specialised to the subtheory $\mathbf{Stoch}$. This
ensures that we can consistently view these convex combinations as describing
our classical uncertainty about which process is happening.

\subsubsection{\textbf{Tomography}}\label{sec:tomog}
The next requirement that a GPT must satisfy is to have a notion of tomography
\cite{hardy2012limited}. What that means is that we should be able to
characterize its elements -- i.e., the states, effects and transformations --
by the statistics that they are capable of generating. In this way, an
experimentalist would be able to characterize the theoretical objects
describing their experiment by connecting the statistics to the probabilities
that the theory predicts.

To have a notion of tomography of processes, we need to always be able to
establish equalities between them by looking at the statistics that they can
generate. In a GPT, this means the following: we require that if it is the case
that whenever we swap the process $f: A \to B$ by the process $g: A \to B$ in
any diagram that represents a stochastic map, that map is kept unchanged, then
it must be that $f = g$. That is,
\begin{equation}\label{eq:16}
\InputIfFileExists{Diagrams/f.tikz}{}{\input{./figures/Diagrams/f.tikz}}\ \ =\ \ %
\InputIfFileExists{Diagrams/g.tikz}{}{\input{./figures/Diagrams/g.tikz}}\quad \iff\quad
	\forall \tau,\ X,\ Y, \ \ %
\InputIfFileExists{Diagrams/tom1.tikz}{}{\input{./figures/Diagrams/tom1.tikz}}\ \  =\ \
\InputIfFileExists{Diagrams/tom2.tikz}{}{\input{./figures/Diagrams/tom2.tikz}}\,.
\end{equation}
Here we are using $\tau$ to represent an arbitrary diagram that, after
inserting $f$ in some specific spot thereof, has only classical inputs and
outputs left, and so is a process in $\mathbf{Stoch}$. Note that this includes
the case where any of the input/output wires of $\tau$ are the trivial system,
because the trivial system is a classical (that is, $\mathbf{Stoch}$) system.
This condition can be phrased in the following way: two processes $f$ and $g$
from $A$ to $B$ are equal (left hand side of Eq.~\eqref{eq:16}) if and only if
they are \textit{operationally equivalent} (right hand side of
Eq.~\eqref{eq:16}).

\subsubsection{\textbf{Causality}}\label{sec:discard}
In this work we are interested in GPTs that are causal
\cite{chiribella2010probabilistic,coecke2014terminality}. By that, we mean that
for each system type $A$, there is a unique effect that we can think of as
discarding, or simply ignoring, a given system. This property is called
causality because it can be used to impose compatibility of the GPT with a
relativistic causal structure \cite{kissinger2017equivalence}. When the theory
satisfies causality, we use a special diagram to denote the unique (for each
system $A$) discarding effect:
\begin{equation}\label{}
\begin{tikzpicture}
	\begin{pgfonlayer}{nodelayer}
		\node [style=upground] (10) at (0, 0.5) {};
		\node [style=none] (11) at (0, 0.25) {};
		\node [style=none] (12) at (0, -0.5) {};
		\node [style=label] (13) at (-0.25, -0.25) {$A$};
	\end{pgfonlayer}
	\begin{pgfonlayer}{edgelayer}
		\draw [qWire] (11.center) to (12.center);
	\end{pgfonlayer}
\end{tikzpicture}
} \,.
\end{equation}
Note that the uniqueness of the discarding effects is given for each fixed
system type. In particular, this means that for composite systems the
discarding is obtained by parallel composition of the discarding of the
subsystems:
\begin{equation}\label{}
\InputIfFileExists{Diagrams/discardab1.tikz}{}{\input{./figures/Diagrams/discardab1.tikz}}\ \ =\ \ %
\InputIfFileExists{Diagrams/discardab2.tikz}{}{\input{./figures/Diagrams/discardab2.tikz}} \,.
\end{equation}
The discarding effects will be used in the next section to define
non-signalling channels for a general GPT, just like the trace is in quantum
theory.

The fact that there is a unique effect immediately means that all of the
processes are discard-preserving \cite{coecke2014terminality}:
\begin{defn}[Deterministic, or Discard-Preserving,
		Process]\label{deterministic}
	A process $f: A \to B$ is deterministic if it is discard preserving,
	that is,
	\begin{equation}\label{}
\InputIfFileExists{Diagrams/dispres.tikz}{}{\input{./figures/Diagrams/dispres.tikz}}\ \ =\ \
\begin{tikzpicture}
	\begin{pgfonlayer}{nodelayer}
		\node [style=upground] (10) at (0, 0.5) {};
		\node [style=none] (11) at (0, 0.25) {};
		\node [style=none] (12) at (0, -0.5) {};
		\node [style=label] (13) at (-0.25, -0.25) {$A$};
	\end{pgfonlayer}
	\begin{pgfonlayer}{edgelayer}
		\draw [qWire] (11.center) to (12.center);
	\end{pgfonlayer}
\end{tikzpicture}
}\,.
	\end{equation}
\end{defn}
In quantum theory, since discarding is the trace operation, this corresponds to
the trace-preserving property. That is, the formalism that we are using here is
the analogue of working with only CPTP maps rather than working with CPTNI
maps. Typically, CPTNI maps are used to describe the potential outcomes of some
measurement, we can instead equally well work only with CPTP maps, by instead
considering all possible outcomes at once, and keeping track of which outcome
occurred by means of an auxiliary classical system.

\subsubsection{\textbf{Local tomography}} \label{sec:localtomog}
In this work, we are interested in GPTs that satisfy a stricter notion of
tomography. We require that the tomography of the processes can be done by
evaluating the probabilities produced by local effects, that is, we require our
GPT to satisfy local tomography \cite{hardy2001quantum}. This is expressed
diagrammatically by the following:
\begin{equation}\label{}
\InputIfFileExists{Diagrams/f.tikz}{}{\input{./figures/Diagrams/f.tikz}}\ \ =\ \ %
\InputIfFileExists{Diagrams/g.tikz}{}{\input{./figures/Diagrams/g.tikz}}\quad \iff\quad
	\forall
	s,\ M,\ Y,\ \ %
\InputIfFileExists{Diagrams/ltom1.tikz}{}{\input{./figures/Diagrams/ltom1.tikz}}\ \ =\ \ %
\InputIfFileExists{Diagrams/ltom2.tikz}{}{\input{./figures/Diagrams/ltom2.tikz}}
\end{equation}
where $s$ is an arbitrary state of $A$, $Y$ is an arbitrary classical system,
and $M$ is an arbitrary measurement of $B$.
Note that this is taking a particular, less general, shape for $\tau$ in the
definition of tomography.

\begin{rem}[]
	A very convenient fact about  locally-tomographic  GPTs is that they
	are
	all subtheories of  $\mathds{R}\mathbf{Linear}$
	\cite{hardy2011reformulating,chiribella2016quantum,schmid2020structure}
	(Example \ref{def:RLin} in Appendix \ref{ap:SMC}),  in the sense that
	all of
	the processes of the former are in the latter (or more rigorously,
	there is an
	injective map between their processes and system types), and they
	compose
	according to $\mathds{R}\mathbf{Linear}$ compositional rules.  This
	will come
	in handy, as in our construction we will use the fact that our GPT is
	one of
	$\mathds{R}\mathbf{Linear}$'s subtheories to write its processes in a
	mathematically concrete way. In particular, both classical and quantum
	theory
	satisfy local tomography, and therefore are  also  subtheories of
	$\mathds{R}\mathbf{Linear}$.
\end{rem}

\bigskip

Now that we  are done discussing the structure of the generalised
probabilistic theories, we can proceed and focus on the properties of the
processes that we are interested in investigating inside those theories.
Namely, we can talk about the non-signalling channels.

\section{Channels in generalised probabilistic theories}\label{nschannels}

In this section we discuss, in the context of generalised probabilistic
theories, the two classes of channels of interest for this paper: the
non-signalling channels, and the common-cause channels (which form a subset of
the non-signalling channels, as we will see).

\subsection{Non-signalling channels}
A practical starting point to understand  what	non-signalling channels in GPTs
are  is to  remind ourselves of  what they are in quantum or classical theory.

Quantum channels are formally completely-positive trace-preserving maps on
density matrices, and specify ways in which quantum systems can be transformed.
The properties of quantum channels are widely studied in the literature
\cite{nielsen2000quantum}, and of particular interest are the quantum channels
that satisfy a form of the no-signalling principle \cite{popescu1994quantum},
introduced first by Beckman, Gottesman, Nielsen, and Preskill
\cite{beckman2001causal} in bipartite setups. These non-signalling  quantum
channels are sometimes referred to as ``causal channels''
\cite{schumacher2005locality}, and do not permit superluminal quantum (nor
classical) communication between two parties -- i.e., two wings of the
experiment. Non-signalling channels were discussed in the context of
multipartite setups by Schumacher and Westmoreland
\cite{schumacher2005locality}.

In general theories -- not necessarily quantum or classical -- one can also
define the concept of a channel as a transformation in the theory that is
discard-preserving (Def.~\ref{deterministic} ), that is, one that preserves, on
any state, the result of the application of the discarding process. In this
context, we can talk about the property of a channel being non-signalling.
In this section, we present a convenient definition of non-signalling channels
in the diagrammatic language that we presented in Sec.~\ref{GPTs}.
Specifically, we want to diagrammatically represent the idea that no
information can flow between the parties.
Consider, for example, a bipartite process $\Lambda: A \otimes B \to C \otimes
	D$. If by discarding system $C$ the resulting process $A \otimes B \to
	D$ is
such that changing system $A$ does not produce any changes in system $D$, then
$\Lambda$ cannot signal from the $AC$ wing of the experiment to the $BD$ wing
of the experiment. In other words, we say that $ \Lambda: A \otimes B \to C
	\otimes D$ is non-signalling from $AC$ to $BD$ if and only if
\begin{equation}\label{}
\InputIfFileExists{Diagrams/lambdaab.tikz}{}{\input{./figures/Diagrams/lambdaab.tikz}}\ \ =\ \ %
\InputIfFileExists{Diagrams/lambdaabb.tikz}{}{\input{./figures/Diagrams/lambdaabb.tikz}}  ,
\end{equation}
where $ \Lambda_{b} : B \to D$ is a valid channel within the theory
\cite{coecke2014terminality}. Note, in particular, that this implies that the
application of any deterministic process (Def. \ref{deterministic}) in the $AC$
wing does not change the marginal channel $\Lambda_{b}$:
\begin{equation}\label{}
\InputIfFileExists{Diagrams/flambdaab.tikz}{}{\input{./figures/Diagrams/flambdaab.tikz}}\ \ =\ \ %
\InputIfFileExists{Diagrams/flambdaabb.tikz}{}{\input{./figures/Diagrams/flambdaabb.tikz}}\ \
	=\ \ %
\InputIfFileExists{Diagrams/lambdaabbp.tikz}{}{\input{./figures/Diagrams/lambdaabbp.tikz}} \,,
\end{equation}
hence, no information can flow from the $AC$ wing to the $BD$ wing of the
experiment. A channel is then said to be non-signalling when it satisfies that
property in both directions between the wings of the experiment.

So far we have presented the case of bipartite non-signalling channels, but the
notion of a multipartite non-signalling channel has also been defined in the
literature \cite{schumacher2005locality}. Here we present a convenient
diagrammatic definition of multipartite non-signalling channels. In order to
define the multipartite generalisation of this condition we need a convenient
way to represent discarding an arbitrary subset of the outputs. To see why,
suppose that $ \Lambda$ is a tripartite channel. If we want to guarantee that
no information can flow from any of the subsystems to any other, we need to
have that
\begin{eqnarray}
\InputIfFileExists{Diagrams/NSTri1.tikz}{}{\input{./figures/Diagrams/NSTri1.tikz}}\ \ &= %
\InputIfFileExists{Diagrams/NSTri2.tikz}{}{\input{./figures/Diagrams/NSTri2.tikz}},\qquad
\InputIfFileExists{Diagrams/NSTri3.tikz}{}{\input{./figures/Diagrams/NSTri3.tikz}}\ \ =\ \ %
\InputIfFileExists{Diagrams/NSTri4.tikz}{}{\input{./figures/Diagrams/NSTri4.tikz}}, \qquad
\InputIfFileExists{Diagrams/NSTri5.tikz}{}{\input{./figures/Diagrams/NSTri5.tikz}}\ \ =\ \ %
\InputIfFileExists{Diagrams/NSTri6.tikz}{}{\input{./figures/Diagrams/NSTri6.tikz}},\\
	&%
\InputIfFileExists{Diagrams/NSTri7.tikz}{}{\input{./figures/Diagrams/NSTri7.tikz}}\ \ = %
\InputIfFileExists{Diagrams/NSTri8.tikz}{}{\input{./figures/Diagrams/NSTri8.tikz}},\quad ...
\end{eqnarray}
and so on. It is easy to see that this can become quite complex quickly as we
increase the number of parties. In order to capture this in a succinct
diagrammatic form, we need a notation which allows us to describe discarding an
arbitrary subset of the outputs (or inputs), for this purpose we first
introduce a bipartitioning processes as follows:

\begin{defn}[Bipartitioning processes $B(K)$]
	Given a set $M = \{1,...,m\}$ take a labelled subset $K = \{ k_{1},
		...,
		k_{n}\} \subset M$ and its complement $ \overline{K} = \{
		\overline{k}_{1},...,
		\overline{k}_{n'}\} = M \setminus K$, where $n+n'=m$.  Then,
	the
	bipartitioning process $B(K)$ is the permutation  which takes
	$(1,...,m)$ to
	$(k_{1},...,k_{n}, \overline{k}_{1},..., \overline{k}_{n'})$.
	Diagrammatically, we represent this by
	\begin{equation}\label{}
		B(M|K)\ \ \doteq\ \ %
\InputIfFileExists{Diagrams/pK.tikz}{}{\input{./figures/Diagrams/pK.tikz}} \,.
	\end{equation}
	where we are using numbers, instead of system type names, to refer to
	the wires
	for the sake of clarity.
	\hfill	$\CIRCLE$
\end{defn}
For example, if we take $M = \{1,2,3\}$ and $K = \{2,3\}$, or $M' =
	\{1,2,3,4\}$ and $K' = \{1,4\}$ then we have, respectively,
\begin{equation}\label{}
\InputIfFileExists{Diagrams/pK2.tikz}{}{\input{./figures/Diagrams/pK2.tikz}}\ \ =\ \ %
\InputIfFileExists{Diagrams/pK3.tikz}{}{\input{./figures/Diagrams/pK3.tikz}} \quad \text{and}
	\quad %
\InputIfFileExists{Diagrams/pK4.tikz}{}{\input{./figures/Diagrams/pK4.tikz}}\ \ =\ \ %
\InputIfFileExists{Diagrams/pK5.tikz}{}{\input{./figures/Diagrams/pK5.tikz}}.
\end{equation}

We can then use this bipartitioning operation to concisely notate discarding
some subset$K$ of the outputs $M$ of a channel $\Lambda$, i.e.,
\begin{equation}\label{}
\InputIfFileExists{Diagrams/ptrace.tikz}{}{\input{./figures/Diagrams/ptrace.tikz}} \,,
\end{equation}
which in quantum theory would represent the partial trace $\tr_{K} (\Lambda)$,
up to a permutation of the surviving systems. For example, in the tripartite
case we can represent discarding the second and third outputs by
\begin{equation}\label{}
\InputIfFileExists{Diagrams/ptrace2.tikz}{}{\input{./figures/Diagrams/ptrace2.tikz}}\ \ =\ \ %
\InputIfFileExists{Diagrams/ptrace3.tikz}{}{\input{./figures/Diagrams/ptrace3.tikz}}\ \ =\ \
\InputIfFileExists{Diagrams/ptrace4.tikz}{}{\input{./figures/Diagrams/ptrace4.tikz}}.
\end{equation}

We can now present the definition of multipartite non-signalling channels in a
succinct  diagrammatic form:
\begin{defn}[Non-signalling channel]\label{nsdef}
	An m-partite channel $ \Lambda: \otimes_{i=1}^{m} i \to \otimes_{i' =
			1}^{m} i' $ is non-signalling iff for all labelled
	subsets $ K \subset
		\{1,...,m\}$, there exists a channel $ \Lambda_{ \overline{K}}
		:  \otimes_{i =
			1}^{n'} \overline{k}_{i} \to \otimes_{ i=1 }^{n'}
		\overline{k}_{i}' $, with
	$\overline{K} = \{1, \ldots,m\} \setminus K$,
	such that
	\begin{equation}\label{eq:NS25}
\InputIfFileExists{Diagrams/NS1.tikz}{}{\input{./figures/Diagrams/NS1.tikz}}\ \ =\ \ %
\InputIfFileExists{Diagrams/NS2.tikz}{}{\input{./figures/Diagrams/NS2.tikz}}.
	\end{equation}
	\hfill	$\CIRCLE$
\end{defn}
To illustrate this, one of the conditions that this definition would impose on
the tripartite case ($M = \{1,2,3\}$) would be for $K = \{2,3\}$, which would
give
\begin{eqnarray}
\InputIfFileExists{Diagrams/NSExNew1.tikz}{}{\input{./figures/Diagrams/NSExNew1.tikz}}\ \ &=\ \ %
\InputIfFileExists{Diagrams/NSEx2.tikz}{}{\input{./figures/Diagrams/NSEx2.tikz}}\ \ =\ \
\InputIfFileExists{Diagrams/NSEx1.tikz}{}{\input{./figures/Diagrams/NSEx1.tikz}}\\ &= \ \ %
\InputIfFileExists{Diagrams/NSEx5.tikz}{}{\input{./figures/Diagrams/NSEx5.tikz}}\ \  =\ \
\InputIfFileExists{Diagrams/NSEx4.tikz}{}{\input{./figures/Diagrams/NSEx4.tikz}}\ \ = \ \  %
\InputIfFileExists{Diagrams/NSExNew2.tikz}{}{\input{./figures/Diagrams/NSExNew2.tikz}}
\end{eqnarray}
that is, we can see explicitly how our condition gives us no signalling from
$2\otimes 3$ to $1$. It is straightforward to similarly verify that the other
conditions in the tripartite case are recovered by varying over the subsets
$K\subseteq M$.

Notice that this definition of a non-signalling channel treats each pair of
input/output systems $(i,i')$ as a different wing of the experiment. Therefore,
when specifying the experimental scenario and the channel $\Lambda$ the systems
should be represented via `one wire per wing'. As an example, consider the case
where one wing of the experiment consists of  two qubits forming a
4-dimensional quantum system as an input: then this must be represented by one
4-dimensional system -- rather than by two wires representing two qubits --
when Definition \ref{nsdef} is applied, since signalling is allowed between the
wing's	internal  two qubits.

\subsection{Common cause channels}

To formally state the question tackled in this paper, we first need to	specify
the notion of a common-cause channel  that we use in this manuscript.	Broadly
speaking, the  common-cause channels are a subset of the non-signalling
channels. Namely, we say that a channel is common-cause if, in the GPT of
interest, it can be constructed by the parties via the application of some
local operations to a shared multipartite state.
A good example of such a channel is the one obtained in a Bell experiment,
where, for example, Alice and Bob each make measurements on their shares of a
Bell state. One can view the result of the Bell experiment as being a bipartite
classical channel which is realised by local operations on a shared quantum
state, i.e., a quantum common-cause.

Based on this example, we can define the notion of a common-cause decomposition
within a given GPT $\mathbf{G}$.
\begin{defn}[Common-cause decomposition]\label{ccddef}
	Let $ \Lambda$ be a channel in a given GPT $\mathbf{G}$. $ \Lambda$
	admits of
	a \textit{common-cause decomposition} if there are $N$ systems
	$\{1''\,,\,
		\ldots \,,\, N''\}$ from $\mathbf{G}$, a state $s$ in the state
	space of the
	multipartite system $1''\,,\, \ldots \,,\, N''$ and a collection
	$\{T_{i}\}_{i=1\ldots N}$ of  transformations in $\mathbf{G}$, such
	that
	\begin{equation}\label{cc}
\InputIfFileExists{Diagrams/lambda1n.tikz}{}{\input{./figures/Diagrams/lambda1n.tikz}}\ \ =\ \
\InputIfFileExists{Diagrams/lambda1nd.tikz}{}{\input{./figures/Diagrams/lambda1nd.tikz}}.
	\end{equation}
	\hfill	 $\CIRCLE$
\end{defn}
One can compare this formal diagrammatic definition to the conceptual
definition to see that indeed the idea of construction by local operations (the
transformations $T_i$) on a shared common cause (the state $s$) is indeed
captured by this diagram.

Now, the idea of common-cause decomposition within a GPT might not be enough if
one is considering the possible existence of some hypothetical cause that might
not be modeled by the GPT under consideration. In particular, this is precisely
the kind of situation that is considered in the resource theories of
Refs.~\cite{cowpie,schmid2020postquantum}. In such cases, the more appropriate
question is not whether $\Lambda$ can be realised with a common cause in
$\mathbf{G}$, but whether or not there exists a theory $\mathbf{G}'$ in which
it can be realised with a common cause.
Going back to our example of the Bell experiment, if we violate a Bell
inequality, then we know that the resulting channel cannot be realised via
common case within $\mathbf{Stoch}$, but it can be realised via a quantum
common cause, that is, within the quantum GPT, $\mathbf{Quant}$.

For that purpose, we  define the notion of \textit{GPT-common-cause
	realisable}, by asking whether the common-cause decomposition of $
	\Lambda$
exists in \textit{any} GPT.
\begin{defn}[GPT-Common-cause realisable channel]\label{commoncausedef}
	Let $ \Lambda$ be a channel in a given GPT $\mathbf{G}$. $\Lambda$ is
	\textit{GPT-common-cause realisable} if there exists a GPT
	$\mathbf{G}'$ which
	contains $\mathbf{G}$ as a full subtheory, $N$ systems $\{1''\,,\,
		\ldots \,,\,
		N''\}$ from $\mathbf{G}'$, a state $s$ in the state space of
	the multipartite
	system $1''\,,\, \ldots \,,\, N''$ in $\mathbf{G}'$, and a collection
	$\{T_{i}\}_{i=1\ldots N}$ of  transformations in $\mathbf{G}'$, such
	that
	\begin{equation}\label{commoncause}
\InputIfFileExists{Diagrams/lambda1n.tikz}{}{\input{./figures/Diagrams/lambda1n.tikz}}\ \ =\ \
\InputIfFileExists{Diagrams/lambda1nd2.tikz}{}{\input{./figures/Diagrams/lambda1nd2.tikz}},
	\end{equation}
	where we changed the colors of the $i''$ wires to stress the fact that
	they
	can be present only in the hypothetical GPT $\mathbf{G}'$, whilst the
	wires $i$
	and $i'$ are required to live in the original subtheory $\mathbf{G}$.
	\hfill	 $\CIRCLE$
\end{defn}

Common-cause realisable channels are well known to be non-signalling, here we
present this result using the diagrammatic notation that we have set up so far.
\begin{prop}\label{cccns}
	Any GPT-common-cause realisable channel is non-signalling.
\end{prop}

\begin{proof}
	Consider a fixed but arbitrary channel $\Lambda$ in a GPT $\mathbf{G}$.
	Let
	$\mathbf{G}'$ be the GPT that provides the common-cause realisation of
	$\Lambda$. First, notice that, because in Eq.~\ref{eq:NS25} the $ T_{
				i} $
	channels are discard-preserving, if we take $ \Lambda$ to be decomposed
	as in
	Eq.~\ref{commoncause}, we get
	\begin{equation}\label{ccnsproof}
		\begin{aligned}
\InputIfFileExists{Diagrams/CCProof1.tikz}{}{\input{./figures/Diagrams/CCProof1.tikz}}\ \	 & =\ \
\InputIfFileExists{Diagrams/CCProof2.tikz}{}{\input{./figures/Diagrams/CCProof2.tikz}}\ \
			=\ \ %
\InputIfFileExists{Diagrams/CCProof3.tikz}{}{\input{./figures/Diagrams/CCProof3.tikz}}\ \ =\ \
\InputIfFileExists{Diagrams/CCProof4.tikz}{}{\input{./figures/Diagrams/CCProof4.tikz}}                        \\
			                               & \hspace{-2cm}=\ \
\InputIfFileExists{Diagrams/CCProof4_2.tikz}{}{\input{./figures/Diagrams/CCProof4_2.tikz}}\ \ =\ \
\InputIfFileExists{Diagrams/CCProof5_2.tikz}{}{\input{./figures/Diagrams/CCProof5_2.tikz}}\ \ =:\ \
\InputIfFileExists{Diagrams/CCProof6.tikz}{}{\input{./figures/Diagrams/CCProof6.tikz}}.
		\end{aligned}
	\end{equation}
	where $ \Lambda_{ \overline{K}} $, the channel defined by combining the
	elements within the dashed box, must be a valid channel from
	$\mathbf{G}$
	because all of its inputs and outputs are from $\mathbf{G}$, and
	$\mathbf{G}$
	is assumed to be a full subtheory of $\mathbf{G}'$.
\end{proof}

The main aim of the paper is hence  to explore	the converse direction to
Proposition \ref{cccns}, namely, whether non-signalling channels can  in
general  be common-cause realisable. The first observation to make is the
well-known fact that the non-signalling classical channel known as
Popescu-Rohrlich (PR) box\footnote{The PR box can be thought of as a
	non-signalling classical channel that takes classical systems to
	classical
	systems.} \cite{popescu1994quantum} does not have a common-cause
realisation
within classical theory \cite{Bell64}, but it does have one such realisation
within the GPT known as Boxworld \cite{barrett2007gpt}. In this sense, hence,
we say that the classical GPT is \emph{common-cause incomplete}. Moreover, we
further view Boxworld as adding extra common causes to classical theory, and so
can be thought of as a \textit{common-cause completion} of classical theory.
This discussion motivates the following definition:
\begin{defn}[Common-cause complete GPT]
	A GPT is said to be \textit{common-cause complete} if a common-cause
	decomposition can be found for each of its non-signalling channels
	within the
	theory. That is, given a non-signalling channel $ \Lambda$ in the GPT,
	we can
	decompose it as in Definition \ref{commoncausedef} taking
	$\mathbf{G}'=\mathbf{G}$. \hfill $\CIRCLE$
\end{defn}

The previous observation shows that there are some GPTs -- such as classical
and quantum theory -- which are not common-cause complete. However, classical
theory does have a common-cause completion.
The question we therefore ask is whether or not this is generic? That is:
\begin{center}
	Given some GPT $\mathbf{G}$, can we find a common-cause completion
	$\mathbf{G}'$ such that all of the non-signalling channels of
	$\mathbf{G}$ have
	a GPT common-cause realisation in $\mathbf{G}'$
	(Def.~\ref{commoncausedef})?
\end{center}
Formally, we defined the common-cause completion as follows:
\begin{defn}[Common-cause completion]
	A GPT $ \mathbf{G}'$ is a common-cause completion of a GPT $
		\mathbf{G}$ if $\mathbf{G}$ is a subtheory of $\mathbf{G}'$,
	and $\mathbf{G}'$
	contains a common-cause decomposition (as per Definition
	\ref{commoncausedef})
	of all of the non-signalling channels of $\mathbf{G}$.
	Note that this definition does not require $\mathbf{G}'$ to be
	common-cause complete  itself.
	\hfill $\CIRCLE$
\end{defn}
In the following section we show that any tomographically-local GPT does indeed
have a common-cause completion.

\section{Common-cause completion}\label{results}

In this section, we provide a construction $\mathcal{C}$ which takes an
arbitrary locally-tomographic causal GPT $\mathbf{G}$ into a common-cause
completion thereof, $\mathcal{C}[\mathbf{G}]$. The starting point of our
construction relies on the following Lemma, proven in
Ref.~\cite{cavalcanti2022decomposing}.
\begin{lem}(Affine common-cause decomposition of non-signalling channels
	\cite{cavalcanti2022decomposing})\label{hermitianstates}
	In a locally-tomographic GPT $\mathbf{G}$, any $m$-partite
	non-signalling
	channel, $\Lambda $, can be written as:
	\begin{equation}\label{lemma1eq}
\InputIfFileExists{Diagrams/NoSignallingChannel.tikz}{}{\input{./figures/Diagrams/NoSignallingChannel.tikz}} \ \ = \ \
\InputIfFileExists{Diagrams/NoSignallingChannelDecomposition.tikz}{}{\input{./figures/Diagrams/NoSignallingChannelDecomposition.tikz}}
	\end{equation}
	where $ \tilde{\eta}^\Lambda_i$ are discard-preserving processes in
	$\mathbf{G}$, and $\tilde{\xi}^\Lambda$ is an affine combination of
	states from
	$\mathbf{G}$ (e.g., when $\mathbf{G}$ is quantum theory,
	$\tilde{\xi}^\Lambda$
	is a unit trace Hermitian operator). Note that we have drawn
	$\tilde{\xi}^{\Lambda}$ as a black box to indicate that, whilst it is a
	mathematically valid object, it is not necessarily  a physical process
	within
	the GPT $\mathbf{G}$\footnote{That is, when we view the tomographically
		local
		theory as a subtheory of $\mathds{R}\mathbf{Linear}$ we can
		then take arbitrary
		affine (or more generally linear) combinations and have a well
		defined process
		in $\mathds{R}\mathbf{Linear}$ but this then could be outside
		of the subtheory
		$\mathbf{G}$.}.
\end{lem}
\proof Theorem 5.1 of Ref.~\cite{cavalcanti2022decomposing}. \endproof

This lemma, at first glance, provides a  common-cause realisation of any
non-signalling channel. However, these affine combinations of states
$\widetilde{\xi}^\Lambda$ are not (in general) going to be valid states in the
GPT. One route to a solution could therefore be to define a
common-cause-completion by enlarging the state-space so that it now includes
these non-physical states. The problem with this approach, however, is that it
does not necessarily yield a well-defined GPT, since this procedure will often
lead to negative probabilities for measurement outcomes when we start composing
these states in ways other than the diagram described in Eq.~\eqref{lemma1eq}.

In order to prevent negative numbers from arising, then, one can by fiat forbid
certain `undesired' compositions. That is, one needs to equip the produced
theory with \textit{restrictions} on how the processes may be composed --
type-matching conditions would no longer be a sufficient compositional
criterion. Such a theory is, in the language of
Ref.~\cite{barrett2019computational}, called a ``non-free'' GPT as one is not
free to compose processes solely based on their system types. Whilst
mathematically consistent, we find it difficult to justify such restrictions on
physical grounds, and hence we will not pursue its study further in this paper.
In what follows, we instead provide a construction of a valid common-cause
completion map, which, given a causal tomographically local GPT will always
build a valid GPT, where composition precisely follows the GPT rules as per
section \ref{GPTs}.

\subsection{Constructing the $\mathcal{C}$ map}\label{constructing}

Here we define a common-cause completion map, $\mathcal{C}$ which takes an
arbitrary tomographically-local GPT, $\mathbf{G}$, as an input and then
constructs a common-cause completion of it, $\mathcal{C}[\mathbf{G}]$, which is
its output.
The basic idea of this construction is to include all the non-physical states $
	\tilde{\xi}^{\Lambda}$ and $\tilde{\eta}_i^\Lambda$ from Lemma
\ref{hermitianstates}, but now with the caveat that the output systems of each
$ \tilde{\xi}^{\Lambda}$ (and consequently the inputs to the
$\tilde{\eta}_i^\Lambda$) are taken to be new system types which are added to
the theory. It will then be the type matching constraints (which are part of
the basic definition of a GPT) which will prevent negative probabilities from
arising when freely composing processes. It is not immediately clear, however,
whether having done so we satisfy all of the other conditions of a GPT, and
indeed this turns out not to be the case. Therefore, some extra steps are
needed in the construction, in particular, to ensure that the theory is convex
and tomographic.

In more detail, the steps followed in the construction, along  with what they
aim to achieve and how we denote them, are the following:

\begin{enumerate}
	\item  Take the non-signalling channels in $\mathbf{G}$ and decompose
	      them
	      as per Lemma \ref{hermitianstates}. Take each output system of
	      each
	      $\tilde{\xi}^\Lambda$ and promote it to a new primitive system
	      type. Collect
	      all these new system types and, together with the system types
	      from
	      $\mathbf{G}$, define a new set of systems types including them
	      all.  Moreover,
	      include as processes within the theory all of the processes from
	      $\mathbf{G}$
	      together with all processes which are required such that these
	      new systems can
	      realise the common-cause channels as per
	      Lemma~\ref{hermitianstates}.
	      \begin{enumerate}
		      \item[Aim:] \textbf{To ensure that the common-cause
				      decompositions}
			      for non-signalling channels of $\mathbf{G}$ exist
			      in $
				      \mathcal{C}[\mathbf{G}]$.
		      \item[Notation:] $ \mathbf{G} \mapsto \mathbf{G} \sqcup
				      \boldsymbol{\eta}$
	      \end{enumerate}
	\item Take the closure of those systems and processes under
	      composition,
	      and of the processes under convex combinations.
	      \begin{enumerate}
		      \item[Aim:] \textbf{To ensure the compositionality and
				      convexity
				      rules are obeyed}.
		      \item[Notation:] $\mathbf{G} \sqcup \boldsymbol{\eta}
				      \mapsto
				      \mathsf{Conv} [\overline{ \mathbf{G}
						      \sqcup \boldsymbol{\eta}}]$
	      \end{enumerate}
	\item Quotient the theory  $\mathsf{Conv} [\overline{ \mathbf{G} \sqcup
				      \boldsymbol{\eta}}]$  via operational
	      equivalence.
	      \begin{enumerate}
		      \item[Aim:] \textbf{To ensure the theory satisfies
				      tomography}.
		      \item[Notation:] $ \mathsf{Conv} [\overline{ \mathbf{G}
						      \sqcup
						      \boldsymbol{\eta}}]
				      \mapsto  \mathsf{Conv} [\overline{ \mathbf{G} \sqcup
						      \boldsymbol{\eta}}] /
				      \sim $
	      \end{enumerate}
\end{enumerate}
It is this theory that we will define as our common-cause completion, i.e.
$\mathcal{C}[\mathbf{G}] := \mathsf{Conv} [\overline{ \mathbf{G} \sqcup
			\boldsymbol{\eta}}] / \sim$.

As we progress through the steps, we will show	that they do indeed achieve the
stated aim. In the end, we will therefore see that the outcome
$\mathcal{C}[\mathbf{G}]$ of this construction is a valid causal GPT (in
particular, that there are no extra restrictions on composing systems and
processes) and that it is a common-cause completion of $\mathbf{G}$.

In this section  we will be dealing with many system types from different GPTs
(due to the nature of the problem of extending a theory),  and therefore  we
shall use colors to differentiate the wires corresponding to different
theories' system types. The convention we follow is given by the following
table:

\begin{center}
	\begin{tabular}{||c|c||}
		\hline
		System Type                                           & Wire
		Type                                                         \\
		\hline
		System from the classical subtheory, $\mathbf{Stoch}$ & $
\begin{tikzpicture}
	\begin{pgfonlayer}{nodelayer}
		\node [style=none] (0) at (0, 0.75) {};
		\node [style=none] (1) at (0, -0.75) {};
		\node [style=none] (4) at (0, 1.25) {};
		\node [style=none] (5) at (0, -1.25) {};
	\end{pgfonlayer}
	\begin{pgfonlayer}{edgelayer}
		\draw [cWire] (0.center) to (1.center);
	\end{pgfonlayer}
\end{tikzpicture}
}$
		\\
		\hline
		Generic system from the target GPT, $\mathbf{G}$      & $
\begin{tikzpicture}
	\begin{pgfonlayer}{nodelayer}
		\node [style=none] (0) at (0, 0.75) {};
		\node [style=none] (1) at (0, -0.75) {};
		\node [style=none] (2) at (0, 1.25) {};
		\node [style=none] (3) at (0, -1.25) {};
	\end{pgfonlayer}
	\begin{pgfonlayer}{edgelayer}
		\draw [style=qWire] (0.center) to (1.center);
	\end{pgfonlayer}
\end{tikzpicture}
}$
		\\
		\hline
		Extra system to be added to $\mathbf{G}$              & $
\begin{tikzpicture}
	\begin{pgfonlayer}{nodelayer}
		\node [style=none] (0) at (0, 0.75) {};
		\node [style=none] (1) at (0, -0.75) {};
		\node [style=none] (2) at (0, 1.25) {};
		\node [style=none] (3) at (0, -1.25) {};
	\end{pgfonlayer}
	\begin{pgfonlayer}{edgelayer}
		\draw [style=pWire, in=90, out=-90] (0.center) to (1.center);
	\end{pgfonlayer}
\end{tikzpicture}
}$                                    \\
		\hline
		Generic system in the new GPT                         & $
\begin{tikzpicture}
	\begin{pgfonlayer}{nodelayer}
		\node [style=none] (0) at (0, 0.75) {};
		\node [style=none] (1) at (0, -0.75) {};
		\node [style=none] (2) at (0, 1.25) {};
		\node [style=none] (3) at (0, -1.25) {};
	\end{pgfonlayer}
	\begin{pgfonlayer}{edgelayer}
		\draw [style=gWire] (0.center) to (1.center);
	\end{pgfonlayer}
\end{tikzpicture}
}$                                    \\
		\hline
	\end{tabular}
\end{center}

\subsubsection*{\textbf{Step 1 - Add generating system types and processes}}
Starting from $ \mathbf{G}$, for each $ \Lambda$  in $ \mathbf{G}$ decomposed
as in Eq.~\eqref{lemma1eq}, let us define a vector space $ A_{i}^{\Lambda}$
which is isomorphic to $ \tilde{i}^{\Lambda}$ with isomorphism $
	\iota_{i}^{\Lambda} : \tilde{i}^{\Lambda} \to A_{i}^{\Lambda}$. Then,
we define
the following linear maps:

\begin{equation}\label{eq29}
\InputIfFileExists{Diagrams/NewMapsDef1.tikz}{}{\input{./figures/Diagrams/NewMapsDef1.tikz}} \ \ :=\ \
\InputIfFileExists{Diagrams/NewMapsDef2.tikz}{}{\input{./figures/Diagrams/NewMapsDef2.tikz}}\qquad
\end{equation}
and
\begin{equation}\label{eq30}
\InputIfFileExists{Diagrams/NewMapsDef3.tikz}{}{\input{./figures/Diagrams/NewMapsDef3.tikz}}\ \ :=\ \
\InputIfFileExists{Diagrams/NewMapsDef4.tikz}{}{\input{./figures/Diagrams/NewMapsDef4.tikz}}.
\end{equation}
Note that the isomorphisms $ \iota_{i}^{\Lambda}$ and their inverses are
\emph{not} taken to be physically realisable processes within the theory that
we are constructing,  hence, we denote them, as above, with black boxes. We
will, however, take the above composites of them with the $
	\tilde{\eta}_{i}^{\Lambda}$ and $ \tilde{\xi}^{\Lambda}$, to give
$\eta_i^\Lambda$ and $\xi^\Lambda$, to be valid processes in the theory we are
defining,  hence why the left-hand-side of Eqs.~\eqref{eq29} and \eqref{eq30}
are white-coloured boxes.

We therefore obtain the following straightforward corollary of Lemma
\ref{hermitianstates}:
\begin{corollary}
	Any $m$-partite non-signalling channel, $N$, can be written as:
	\begin{equation}\label{}
\InputIfFileExists{Diagrams/NoSignallingChannel.tikz}{}{\input{./figures/Diagrams/NoSignallingChannel.tikz}}\ \ =\ \
\InputIfFileExists{Diagrams/NoSignallingChannelDecompositionGPT.tikz}{}{\input{./figures/Diagrams/NoSignallingChannelDecompositionGPT.tikz}}
	\end{equation}
\end{corollary}
\proof This immediately follows from the definition of the $\eta_i^\Lambda$ and
the $\xi^\Lambda$ (Eqs.~\eqref{eq29} and \eqref{eq30}) together with the fact
that the $\iota_i^\Lambda$ are isomorphisms.

We include these extra systems $ A_{i}^{\Lambda}$ and processes $
	\xi^{\Lambda}$, $ \eta_{i}^{\Lambda}$ within the GPT we are building,
thereby
extending $ \mathbf{G}$ and enabling the realisation of arbitrary
non-signalling channels from $\mathbf{G}$ within the common-cause scenario.

\subsubsection*{\textbf{Step 2 - Take closure under compositions and convex
		combinations}}
For the second step, let us denote by $ |\mathbf{G} |$ the collection of
systems of $ \mathbf{G}$, and (with slight abuse of notation) by $ \mathbf{G}$
the collection of its processes. In order to define the closure properties that
we want, we will note that we can view all of the processes that we have
defined as living within the process theory of real linear maps,
$\mathds{R}\mathbf{Linear}$. To see this, recall that $\mathbf{G}$ is, by
assumption, tomographically local, and hence is a subtheory of
$\mathds{R}\mathbf{Linear}$, and that the new systems and processes that we
have added are all, by definition, real linear maps.

We therefore define another subtheory of $\mathds{R}\mathbf{Linear}$ which is,
by construction, closed under composition as follows:

\begin{defn}[]
	We denote by $ \overline{\mathbf{G} \sqcup \boldsymbol{\eta}} $ the
	subtheory of  $\mathds{R}\mathbf{Linear}$ whose objects (system types)
	are the
	closure of $| \mathbf{G} | \sqcup \{A_{i}^{\Lambda}\}_{\Lambda, i}$
	under $
		\otimes$, and whose morphisms (processes) are the closure of $
		\mathbf{G}
		\sqcup \{\eta_{i}^{\Lambda} , \xi^{\Lambda}\}_{\Lambda, i}$
	under $ \circ$ and
	$ \otimes$ as the operations in $\mathds{R}\mathbf{Linear}$.
	\hfill $\CIRCLE$
\end{defn}
Note that, even though we did not explicitly mention the states of the $
	A_{i}^{\Lambda}$ systems, these are implicitly defined by the above
closure to
obtain $ \overline{\mathbf{G} \sqcup \boldsymbol{\eta}} $. For example, by
varying over $ \rho$ in the following diagram, we can obtain many states of $
	A_{1}^{\Lambda}$:
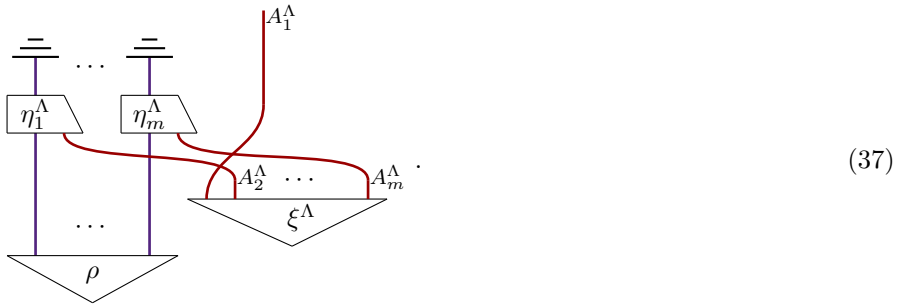
\begin{equation}\label{}
\InputIfFileExists{Diagrams/A1N-States.tikz}{}{\input{./figures/Diagrams/A1N-States.tikz}}.
\end{equation}
In  the same way,  effects and other general processes on the new system types
$ A_{i}^{\Lambda}$ can also be defined. The fact that we only have an implicit
definition of the state and effect space is in stark contrast to traditional
ways of constructing GPTs, in which the convex geometry of the state and effect
spaces is typically the first thing to be defined and then the compositional
structure is built on top of this. Here we invert this, first starting with the
compositional structure and then defining the geometry of the states and
effects which this provides.

Next we will check whether $ \overline{\mathbf{G} \sqcup \boldsymbol{\eta}} $
leads to sensible probabilistic predictions, namely, whether it contains
$\mathbf{Stoch}$ as a full subtheory.
To answer this we note that $\mathbf{Stoch}$ is a full subtheory of
$\mathbf{G}$ and show that $\mathbf{G}$ is a full subtheory of	$
	\overline{\mathbf{G} \sqcup \boldsymbol{\eta}} $, hence, by
transitivity, that
$\mathbf{Stoch}$ is a full subtheory of  $ \overline{\mathbf{G} \sqcup
		\boldsymbol{\eta}} $.

Specifically, what we need to show is that any process with all inputs and
outputs in $|\mathbf{G}|$, such as
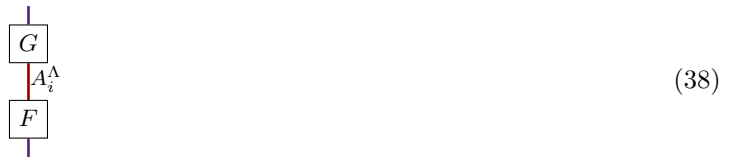
\begin{equation}\label{}
\InputIfFileExists{Diagrams/QuantumExt.tikz}{}{\input{./figures/Diagrams/QuantumExt.tikz}}
\end{equation}
yields a valid process from $\mathbf{G}$. Note that this is not guaranteed
apriori, due to the fact that the new systems $A_i^\Lambda$ appear in the
interior of the diagram. However, in our case it turns out that this is true as
is proven in the following lemma.
\begin{lem}\label{lem:positivity}
	Any process in $\overline{\mathbf{G}\sqcup \boldsymbol{\eta}}$ with
	only input and output system types in $| \mathbf{G} |$ is a valid
	process in $
		\mathbf{G}$.
\end{lem}
\begin{proof}
	The proof can be found in Appendix \ref{se:pl2}.
\end{proof}

Next we show that $\overline{\mathbf{G}\sqcup \boldsymbol{\eta}}$ is compatible
with relativistic causal structure, in the sense that there is a unique effect
for each system \cite{coecke2014terminality,kissinger2017equivalence}.
\begin{lem} \label{lem:causality}
	There is a unique effect for each system in $\overline{\mathbf{G}\sqcup
			\boldsymbol{\eta}}$.
\end{lem}
\begin{proof}
	The proof can be found in Appendix \ref{se:pl3}.
\end{proof}

A GPT must also be closed under convex combinations so as to model
probabilistic mixtures of processes, and so far we have not proven that this is
the case for  $\overline{\mathbf{G}\sqcup \boldsymbol{\eta}}$. Indeed, it  is
conceivable that this property has been lost when adding in the new systems and
processes and arbitrary diagrams thereof.
Hence, we take the convex closure of $ \overline{ \mathbf{G} \sqcup
		\boldsymbol{\eta}}$, via the convex combinations of linear maps
provided by the
supertheory $\mathds{R}\mathbf{Linear}$.

\begin{defn}[]
	We denote by $\mathsf{Conv}[\overline{ \mathbf{G} \sqcup
				\boldsymbol{\eta} }]$ the convex closure of $
		\overline{ \mathbf{G} \sqcup
			\boldsymbol{\eta} }$ under convex combinations of
	processes taken as linear
	combinations of linear maps from $\mathds{R}\mathbf{Linear}$.	 \hfill
	$\CIRCLE$
\end{defn}
Notice that the properties of `has $\mathbf{Stoch}$ as a full subtheory' and
`is causal' that we proved for $ \overline{ \mathbf{G} \sqcup \boldsymbol{\eta}
	}$ are properties which must hold in any GPT, hence we next show they
also
hold for $\mathsf{Conv}[ \overline{ \mathbf{G} \sqcup \boldsymbol{\eta} }]$:
\begin{lem}\label{lem:tomcauspreserv}
	i) Any process in $\mathsf{Conv}[\overline{\mathbf{G}\sqcup
				\boldsymbol{\eta}}]$ with only	input and
	output system types in $| \mathbf{G}
		|$ is a valid process in $ \mathbf{G}$.  ii) There is a unique
	discarding
	effect for each system in  $\mathsf{Conv}[\overline{\mathbf{G}\sqcup
				\boldsymbol{\eta}}]$.
\end{lem}
\begin{proof}
	The proof can be found in Appendix \ref{se:pl4}.
\end{proof}

\subsubsection*{\textbf{Step 3 - Quotient the theory}}

There is one final property which must be satisfied in order to have a GPT on
our hands, that is, tomography. That means that we need to be able to establish
the equality between two processes when the probabilities that they can produce
are the same. At this point, however, we do not know that
$\mathsf{Conv}[\overline{\mathbf{G}\sqcup \boldsymbol{\eta}}]$ satisfies this
property. Hence, we need a way to ``merge'' any two differently-labelled but
operationally-equivalent processes (defined shortly) into a single one.

To enforce this, we simply take the quotient
$\mathsf{Conv}[\overline{\mathbf{G}\sqcup \boldsymbol{\eta}}]$ under
operational equivalence. That amounts to defining processes to be equivalence
classes and also the operations of sequential, parallel, and convex
compositions thereof. For this, let us first formally specify what we mean by
``operational equivalence''.

\begin{defn}Processes $f$ and $f'$ (with the same input systems and the same
	output systems) are operationally equivalent if they give the same
	statistical
	predictions when composed with any circuit fragment $\tau$ such that
	the
	resulting process has only classical inputs and outputs:
	\begin{equation}
\begin{tikzpicture}
	\begin{pgfonlayer}{nodelayer}
		\node [style=small box] (0) at (0, 0) {$f$};
		\node [style=none] (1) at (0, 1) {};
		\node [style=none] (2) at (0, -1) {};
	\end{pgfonlayer}
	\begin{pgfonlayer}{edgelayer}
		\draw [gWire] (2.center) to (0);
		\draw [gWire] (1.center) to (0);
	\end{pgfonlayer}
\end{tikzpicture}
}\ \ \sim \ \ %
\begin{tikzpicture}
	\begin{pgfonlayer}{nodelayer}
		\node [style=small box] (0) at (0, 0) {$f'$};
		\node [style=none] (1) at (0, 1) {};
		\node [style=none] (2) at (0, -1) {};
	\end{pgfonlayer}
	\begin{pgfonlayer}{edgelayer}
		\draw [gWire] (2.center) to (0);
		\draw [gWire] (1.center) to (0);
	\end{pgfonlayer}
\end{tikzpicture}
}
		\quad
		\iff \quad \forall \tau\ \  %
\InputIfFileExists{Diagrams/TomogDef3.tikz}{}{\input{./figures/Diagrams/TomogDef3.tikz}}\ \ = \ \
\InputIfFileExists{Diagrams/TomogDef4.tikz}{}{\input{./figures/Diagrams/TomogDef4.tikz}}
	\end{equation}
	\hfill $\CIRCLE$
\end{defn}
Note that we are using green wires to denote arbitrary systems which may be
$\mathbf{G}$-type, the new systems $A_i^{\Lambda} $, or even systems of the
quotiented theory, because operational equivalence is a concept defined
independently of the theory. In any case, we will apply this here only to
$\mathsf{Conv}[\overline{\mathbf{G} \sqcup \boldsymbol{\eta}}]$ in order to
construct the quotiented theory.
We denote the equivalence classes defined by this by square brackets, hence we
can write that $f\sim f' \iff [f]=[f']$, and moreover think of some $f' \in
	[f]$ as providing a representative for the equivalence class of
operations that
$f'$ belongs to.

In order to build a theory in which processes are labelled by equivalence
classes of processes, we must first define a notion of composition for the
equivalence classes.
\begin{defn}\label{def:QuotComp}
	The equivalence classes of processes compose sequentially as
	\begin{equation}
\InputIfFileExists{Diagrams/QuotComp1.tikz}{}{\input{./figures/Diagrams/QuotComp1.tikz}} \ \ :=\ \
\begin{tikzpicture}
	\begin{pgfonlayer}{nodelayer}
		\node [style=small box] (0) at (0, 0) {$[g\circ f]$};
		\node [style=none] (2) at (0, -1) {};
		\node [style=none] (5) at (0, 1) {};
	\end{pgfonlayer}
	\begin{pgfonlayer}{edgelayer}
		\draw [gWire] (2.center) to (0);
		\draw [gWire] (5.center) to (0);
	\end{pgfonlayer}
\end{tikzpicture}
}
	\end{equation}
	and compose in parallel as
	\begin{equation}
\InputIfFileExists{Diagrams/QuotComp3.tikz}{}{\input{./figures/Diagrams/QuotComp3.tikz}} \ \ := \ \
\begin{tikzpicture}
	\begin{pgfonlayer}{nodelayer}
		\node [style=small box] (0) at (0, 0) {$[g\otimes f]$};
		\node [style=none] (2) at (0, -1) {};
		\node [style=none] (5) at (0, 1) {};
	\end{pgfonlayer}
	\begin{pgfonlayer}{edgelayer}
		\draw [gWire] (2.center) to (0);
		\draw [gWire] (5.center) to (0);
	\end{pgfonlayer}
\end{tikzpicture}
}
	\end{equation}
	\hfill $\CIRCLE$
\end{defn}
For these to be valid operations between equivalence classes, they must
not depend on the choices of representatives:
\begin{lem}
	\label{lem:comp}Composition as defined in Def.~\ref{def:QuotComp} is
	independent of the choices of representatives. That is,
	\begin{equation}
		\begin{aligned}
\begin{tikzpicture}
	\begin{pgfonlayer}{nodelayer}
		\node [style=small box] (0) at (0, 0) {$f$};
		\node [style=none] (1) at (0, 1) {};
		\node [style=none] (2) at (0, -1) {};
	\end{pgfonlayer}
	\begin{pgfonlayer}{edgelayer}
		\draw [gWire] (2.center) to (0);
		\draw [gWire] (1.center) to (0);
	\end{pgfonlayer}
\end{tikzpicture}
}\	\sim\
\begin{tikzpicture}
	\begin{pgfonlayer}{nodelayer}
		\node [style=small box] (0) at (0, 0) {$f'$};
		\node [style=none] (1) at (0, 1) {};
		\node [style=none] (2) at (0, -1) {};
	\end{pgfonlayer}
	\begin{pgfonlayer}{edgelayer}
		\draw [gWire] (2.center) to (0);
		\draw [gWire] (1.center) to (0);
	\end{pgfonlayer}
\end{tikzpicture}
}\quad & \text{and}\quad
\begin{tikzpicture}
	\begin{pgfonlayer}{nodelayer}
		\node [style=small box] (0) at (0, 0) {$g$};
		\node [style=none] (1) at (0, 1) {};
		\node [style=none] (2) at (0, -1) {};
	\end{pgfonlayer}
	\begin{pgfonlayer}{edgelayer}
		\draw [gWire] (2.center) to (0);
		\draw [gWire] (1.center) to (0);
	\end{pgfonlayer}
\end{tikzpicture}
}\	\sim \
\begin{tikzpicture}
	\begin{pgfonlayer}{nodelayer}
		\node [style=small box] (0) at (0, 0) {$g'$};
		\node [style=none] (1) at (0, 1) {};
		\node [style=none] (2) at (0, -1) {};
	\end{pgfonlayer}
	\begin{pgfonlayer}{edgelayer}
		\draw [gWire] (2.center) to (0);
		\draw [gWire] (1.center) to (0);
	\end{pgfonlayer}
\end{tikzpicture}
}\qquad
			\implies  \qquad %
\InputIfFileExists{Diagrams/DPProof1.tikz}{}{\input{./figures/Diagrams/DPProof1.tikz}} \	\sim \
\InputIfFileExists{Diagrams/DPProof2.tikz}{}{\input{./figures/Diagrams/DPProof2.tikz}} \quad \text{and} \quad
\InputIfFileExists{Diagrams/DPProof3.tikz}{}{\input{./figures/Diagrams/DPProof3.tikz}}
			\  \sim\   %
\InputIfFileExists{Diagrams/DPProof4.tikz}{}{\input{./figures/Diagrams/DPProof4.tikz}}.
		\end{aligned}
	\end{equation}
\end{lem}
\begin{proof}
	The proof can be found in Appendix \ref{se:pl5}.
\end{proof}

In a similar way we can define convex combinations of equivalence classes as
follows:
\begin{defn}\label{def:QuotMix} Convex mixtures of equivalence classes of
	processes are given by the following:
	\begin{equation}
		p\ %
\begin{tikzpicture}
	\begin{pgfonlayer}{nodelayer}
		\node [style=small box] (0) at (0, 0) {$[f]$};
		\node [style=none] (2) at (0, -1) {};
		\node [style=none] (5) at (0, 1) {};
	\end{pgfonlayer}
	\begin{pgfonlayer}{edgelayer}
		\draw [gWire] (2.center) to (0);
		\draw [gWire] (5.center) to (0);
	\end{pgfonlayer}
\end{tikzpicture}
} \ \ + \ (1-p) \
\begin{tikzpicture}
	\begin{pgfonlayer}{nodelayer}
		\node [style=small box] (0) at (0, 0) {$[g]$};
		\node [style=none] (2) at (0, -1) {};
		\node [style=none] (5) at (0, 1) {};
	\end{pgfonlayer}
	\begin{pgfonlayer}{edgelayer}
		\draw [gWire] (2.center) to (0);
		\draw [gWire] (5.center) to (0);
	\end{pgfonlayer}
\end{tikzpicture}
}
		\quad := \quad %
\begin{tikzpicture}
	\begin{pgfonlayer}{nodelayer}
		\node [style=small box] (0) at (0, 0) {$[pf+(1-p)g]$};
		\node [style=none] (2) at (0, -1) {};
		\node [style=none] (5) at (0, 1) {};
	\end{pgfonlayer}
	\begin{pgfonlayer}{edgelayer}
		\draw [gWire] (2.center) to (0);
		\draw [gWire] (5.center) to (0);
	\end{pgfonlayer}
\end{tikzpicture}
}
	\end{equation}
	\hfill $\CIRCLE$
\end{defn}
It is easy to see that the relevant properties of convex  combinations, for
example distributivity	over $\circ$ and $\otimes$, are immediately inherited
from the analogous property in the prequotiented theory. Again, for
consistency, we prove the following:
\begin{lem}\label{lem:convmix}
	Convex mixtures as defined in Def.~\ref{def:QuotMix} are independent of
	the
	choice of representative. That is:
	\begin{equation}
\begin{tikzpicture}
	\begin{pgfonlayer}{nodelayer}
		\node [style=small box] (0) at (0, 0) {$f$};
		\node [style=none] (1) at (0, 1) {};
		\node [style=none] (2) at (0, -1) {};
	\end{pgfonlayer}
	\begin{pgfonlayer}{edgelayer}
		\draw [gWire] (2.center) to (0);
		\draw [gWire] (1.center) to (0);
	\end{pgfonlayer}
\end{tikzpicture}
} \sim
\begin{tikzpicture}
	\begin{pgfonlayer}{nodelayer}
		\node [style=small box] (0) at (0, 0) {$f'$};
		\node [style=none] (1) at (0, 1) {};
		\node [style=none] (2) at (0, -1) {};
	\end{pgfonlayer}
	\begin{pgfonlayer}{edgelayer}
		\draw [gWire] (2.center) to (0);
		\draw [gWire] (1.center) to (0);
	\end{pgfonlayer}
\end{tikzpicture}
}\quad \text{and} \quad
\begin{tikzpicture}
	\begin{pgfonlayer}{nodelayer}
		\node [style=small box] (0) at (0, 0) {$g$};
		\node [style=none] (1) at (0, 1) {};
		\node [style=none] (2) at (0, -1) {};
	\end{pgfonlayer}
	\begin{pgfonlayer}{edgelayer}
		\draw [gWire] (2.center) to (0);
		\draw [gWire] (1.center) to (0);
	\end{pgfonlayer}
\end{tikzpicture}
}  \sim
\begin{tikzpicture}
	\begin{pgfonlayer}{nodelayer}
		\node [style=small box] (0) at (0, 0) {$g'$};
		\node [style=none] (1) at (0, 1) {};
		\node [style=none] (2) at (0, -1) {};
	\end{pgfonlayer}
	\begin{pgfonlayer}{edgelayer}
		\draw [gWire] (2.center) to (0);
		\draw [gWire] (1.center) to (0);
	\end{pgfonlayer}
\end{tikzpicture}
}\quad \\
		\implies\\ \quad %
\begin{tikzpicture}
	\begin{pgfonlayer}{nodelayer}
		\node [style=small box] (0) at (0, 0) {$pf+(1-p)g$};
		\node [style=none] (1) at (0, 1) {};
		\node [style=none] (2) at (0, -1) {};
	\end{pgfonlayer}
	\begin{pgfonlayer}{edgelayer}
		\draw [gWire] (2.center) to (0);
		\draw [gWire] (1.center) to (0);
	\end{pgfonlayer}
\end{tikzpicture}
}\ \sim\
\begin{tikzpicture}
	\begin{pgfonlayer}{nodelayer}
		\node [style=small box] (0) at (0, 0) {$pf'+(1-p)g'$};
		\node [style=none] (1) at (0, 1) {};
		\node [style=none] (2) at (0, -1) {};
	\end{pgfonlayer}
	\begin{pgfonlayer}{edgelayer}
		\draw [gWire] (2.center) to (0);
		\draw [gWire] (1.center) to (0);
	\end{pgfonlayer}
\end{tikzpicture}
}.
	\end{equation}
\end{lem}
\begin{proof}
	The proof can be found in Appendix \ref{se:pl6}.
\end{proof}

\bigskip

These operations allow us to define the quotiented theory as follows:
\begin{defn}
	We denote the theory whose processes are operational equivalence
	classes of the processes in $ \mathsf{Conv} [ \overline{ \mathbf{G}
				\sqcup
				\boldsymbol{\eta} }] $, with composition and
	convex mixtures given by Defs.
	\ref{def:QuotComp} and \ref{def:QuotMix}, by  $ \mathsf{Conv}
		[\overline{
				\mathbf{G} \sqcup \boldsymbol{\eta} } ] / \sim
	$ \hfill $\CIRCLE$
\end{defn}

Note that, as $ \mathbf{G}$ is a GPT, and hence satisfies tomography, for a
valid process $f$ in $ \mathbf{G}$, we have $[f] = \{f\}$, that is, each
equivalence class of processes in $ \mathbf{G}$ contains a single element. It
is then clear that Lemma~\ref{lem:positivity} also holds for our quotiented
theory.
Moreover, it is also clear that Lemma~\ref{lem:causality} continues to hold
even in our quotiented theory, as quotienting could only identify effects for a
particular system with one another, and as we only have  a unique effect for a
given system in the first place we have a unique effect after quotienting.

The theory $\mathsf{Conv}[\overline{\mathbf{G}\sqcup \boldsymbol{\eta}}]/\sim$
therefore satisfies all of the desired properties to be considered a causal
GPT.

While the GPT that we constructed is  $\mathsf{Conv}[\overline{\mathbf{G}\sqcup
			\boldsymbol{\eta}}]/\sim$, it is clear that it is much
easier to perform
calculations within $\mathsf{Conv}[\overline{\mathbf{G}\sqcup
			\boldsymbol{\eta}}]$ as it is simply a subtheory of
$\mathds{R}\mathbf{Linear}$. Luckily one can always perform calculations in
$\mathsf{Conv}[\overline{\mathbf{G}\sqcup \boldsymbol{\eta}}]/\sim$ by picking
suitable representative elements for the equivalence classes, doing a
computation within $\mathds{R}\mathbf{Linear}$, and then requotienting to
determine the resultant equivalence class.

\begin{defn}[\textbf{Common-cause completion map}]
	The map $\mathcal{C}$ given by $\mathcal{C}[\mathbf{G}] \equiv
		\mathsf{Conv}[\overline{\mathbf{G}\sqcup
				\boldsymbol{\eta}}]/\sim$ is a
	common-cause completion map on the set of causal locally-tomographic
	GPTs.
	\hfill $\CIRCLE$
\end{defn}

This is because $\mathcal{C}[\mathbf{G}]$ is a valid GPT which contains
$\mathbf{G}$ as a full subtheory and where every $ \Lambda \in \mathbf{G}$ has
a common-cause realisation in $\mathcal{C}[\mathbf{G}]$.

\section{Results and discussion} \label{conclusion}

The construction we have presented for a common-cause completion map  is useful
as it allows us to understand possible causal explanations of physical
phenomena. To elaborate on this, let us first introduce our main theorem and a
useful corollary.

\begin{theorem}
	Given a locally-tomographic causal GPT $\mathbf{G}$, its set of
	multipartite
	non-signalling channels (Def.~\ref{nsdef}) is the same as its set of
	multipartite common-cause realisable (Def.~\ref{commoncausedef})
	channels.
	Notice these common-causes might not be state-preparations allowed in
	$\mathbf{G}$
\end{theorem}
\proof
Consider the GPT $\mathcal{C}[\mathbf{G}]$. By Prop.~\ref{cccns}, the
common-cause realisable channels in $\mathbf{G}$ are non-signalling. In the
other direction, by construction,  $\mathcal{C}[\mathbf{G}]$  can provide a
common-cause realisation of any non-signalling channel of $ \mathbf{G}$.
\endproof
Noting that $\mathbf{Quant}$ is a locally-tomographic causal GPT we immediately
obtain the following:
\begin{corollary}\label{TheCorollary}
	There exists a causal GPT that provides a common-cause realisation of
	every
	non-signalling quantum channel. Such a GPT is given by
	$\mathcal{C}[\mathbf{Quant}] $.
\end{corollary}

This corollary is important for two reasons. Firstly, it answers in the
negative `Open Question 1' posed in Ref.~\cite{schmid2020postquantum}:
\textit{Do there exist bipartite non-signalling quantum channels which cannot
	be realized by GPT common causes?}.

Secondly, recall the phenomenon of Einstein-Podolski-Rosen (EPR) inference
\cite{rossi2022characterising} (a.k.a. steering) where a party (say Alice)
learns about the state preparation of a physical system (held by a distant
party, hereon called Bob) by performing measurements on her share of the
bipartite physical system \cite{schrodinger36,SteeWiseman07}. Here the object
of study is the collection of \textit{subnormalised conditional states} that
Bob's subsystem may be prepared in, usually called an \textit{assemblage}
\cite{pusey2013negativity}. Similarly to the case of non-signalling
correlations in Bell experiments, one may mathematically define general
assemblages as those which comply with the no-signalling principle. Given the
particular causal structure that underpins these EPR experiments, then, a
crucial foundational question is whether these general assemblages could be
realised within some (beyond quantum) GPT as a common-cause process. This
question can be readily answered in the affirmative by Cor.~\ref{TheCorollary},
given that assemblages in EPR scenarios can be formalised in terms of
non-signalling quantum-classical channels \cite{pqsc}. This sets the foundation
stone to be able to study the non-classicality of EPR assemblages based on the
properties of the common-cause process within the GPT that may realise them.
In particular, this observation answers in the affirmative the question posed
in Ref.~\cite{cavalcanti2022post}: there exists a causal GPT $\mathbf{Q}'$ that
provides a common-cause realisation of every general assemblage.

More generally, our result provides the fundamental justification of the
possibility to assess and quantify the non-classicality of arbitrary
non-signalling processes by means of the non-classicality of the common-cause
required to realise  them.  This has  previously been argued at length for the
case of correlations in Bell scenarios \cite{cowpie}, where the existence of
common-cause realisations of non-signalling boxes had already been provided by
the GPT known as Boxworld \cite{barrett2007gpt}. In this light, hence, our work
enables the possibility of extending this causal reasoning to scenarios beyond
Bell experiments, which involve other local systems types rather than strictly
classical ones.

\bigskip

Looking forward, there are many open questions pertaining to the common-cause
completion construction that we defined:
\begin{itemize}
	\item Is $\mathcal{C}[\mathbf{G}]$ common-cause complete? Intuitively it seems
	      that this should be the case, but conceivably there  may be  non-signalling
	      channels between the new systems which are not realisable in common cause
	      scenarios within $\mathcal{C}[\mathbf{G}]$. Note that
	      $\mathcal{C}[\mathcal{C}[\mathbf{G}]]$ may not be well defined because we do
	      not yet know whether or not:
	      \begin{itemize}
		      \item $\mathcal{C}[\mathbf{G}]$ is tomographically local, or
		      \item whether or not there is a way to extend the common-cause completion to
		            tomographically-nonlocal GPTs, or to more general kinds of process theories.
	      \end{itemize}\par\noindent
\end{itemize}\par\noindent
Whilst being of technical nature, we expect the answers to these questions to
also help us deepen our understanding on the possible non-signalling processes
that can be motivated, understood, and studied from the causal perspective.
\section*{Acknowledgements}
P.J.C. and A.B.S. acknowledge support by the Foundation for Polish Science (IRAP project, ICTQT, Contract
No. 2018/MAB/5, co-financed by EU within Smart Growth
Operational Programme). J.H.S. was supported by the National Science Centre, Poland (Opus project, Categorical
Foundations of the Non-Classicality of Nature, Project No.
2021/41/B/ST2/03149). All of the diagrams within this
manuscript were prepared using TikZit.

\bibliographystyle{apsrev4-2}
\bibliography{bibliography}

\appendix

\section{(strict) Symmetric Monoidal Categories}\label{ap:SMC}

Since we define a generalised probabilistic theory (GPT) in terms of a strict
symmetric monoidal category (SMC), we devote this appendix to define the
latter. We follow that with a brief commentary on interpreting that structure
in terms of processes, which is key  for understanding how to see that GPTs are
SMCs, and end with the most important example of SMC for this paper.

A (strict) symmetric monoidal category consists of (i) a collection of objects
$A,B...$, (ii) for each pair of objects $A,B$, a collection of morphisms $f:A
	\to B$, and (iii) two operations, $ \circ$ and $ \otimes$ under which
the
category is closed.
The first operation, $ \circ$, maps certain pairs of morphisms to morphisms. In
particular, it combines $f: A \to B$ and $g:B \to C$ into $g \circ f: A \to C$,
and can be performed only when the domain of $g$ matches the codomain of $f$
(in this example, the matching is given by the object $B$). Furthermore,
$\circ$ is associative, so it is similar to function composition.
(iv) An identity morphism $1_{A} : A \to A$  that is a unit for $\circ$, is
moreover associated to each object $A$.
The second operation, $ \otimes$, combines arbitrary pairs of objects, taking
$A$ and $B$ to $A\otimes B$ as well as arbitrary pairs of morphisms, taking
$f:A \to B$ and $g: C \to D$ into $ f \otimes g: A \otimes C \to B \otimes D$.
Furthermore, $ \otimes $, is associative and has a unit object which we denote
$I$, so it is a monoid operation on the collection of objects, being therefore
responsible for the monoidal structure of the category. Finally, the two
operations satisfy a consistency condition, namely that $(g\circ f) \otimes
	(g'\circ f') = (g\otimes g')\circ (f\otimes f')$.

An interesting property of the symmetric monoidal categories is that they
feature a diagrammatic calculus, which provides an intuitive and expressive way
to write  and perform mathematical  calculations. For a description of that, we
refer  the reader  to the section \ref{diagrams} of the main text.

The bare structure of the SMC has a nice interpretation in terms of processes
\cite{coecke2011universe}. We take the objects to represent system types, and
call the monoidal unit, denoted $I$, the trivial system. The morphisms $f: A
	\to B$ are interpreted as processes that take a system of type $A$ into
a
system of type $B$.
The processes that start (but do not end) in the unit object (the trivial
system), i.e., those like $ s: I \to A$, are called states, the ones that end
(but do not start) in $I$, like $e: A \to I$, are called effects, and the ones
who neither start nor end in $I$, such as $f:A \to B$, are called
transformations.
This is intuitive because $ s: I \to A$ can be viewed as some preparation
procedure of a system of type $A$, and $e: A \to I$ as a destructive operation.
Next, processes that start and end in $I$, such as $p: I \to I$, are called
\textit{numbers}, or scalars.
Now, processes can happen sequentially or in parallel, and this is captured by
the SMC -- we interpret $ g \circ f$ as the sequential composition of the
processes $f$ and $g$, where $f$ is followed by $g$ (which acts on the output
of $f$), and $ f \otimes g$ as the composite process given by $f$ and $g$
occurring in parallel. This interpretation of $\circ$ and $\otimes$  motivates
the consistency condition that they had to satisfy,  since that is  the natural
relationship between processes happening in parallel and in sequence.

We now illustrate this abstract definition  of an SMC by means of  the key
example for this paper:
\begin{example}[$\mathds{R}\mathbf{Linear}$]\label{def:RLin}
	The SMC $\mathds{R}\mathbf{Linear}$ takes objects (system types) to be
	real
	vector spaces, and,  morphisms (processes) to be linear maps between
	the vector
	spaces. The $ \circ$ operation is the composition of linear maps, and $
		\otimes$ is the tensor product. The identity morphisms are
	given by the
	identity linear maps, and the monoidal unit is given by the one
	dimensional
	vector space $\mathds{R}$.
\end{example}

\section{Proofs}

\subsection{Proof of Lemma \ref{lem:positivity}}\label{se:pl2}
\textbf{Lemma \ref{lem:positivity}}.  Any process in
$\overline{\mathbf{G}\sqcup \boldsymbol{\eta}}$ with only  input and output
system types in $| \mathbf{G} |$ is a valid process in $ \mathbf{G}$.
\begin{proof}
	First note that, by definition, any process in
	$\overline{\mathbf{G}\sqcup
			\boldsymbol{\eta}}$ can always be written as a diagram
	involving only our
	generating processes, that is, processes in $\mathbf{G}$, and the
	processes  in
	$\{\eta_i^{\Lambda} , \xi^{\Lambda} \}_{\Lambda, i}$.

	Now consider an arbitrary process  $F$ in $\overline{\mathbf{G}\sqcup
			\boldsymbol{\eta}}$ with input and output system types
	in $| \mathbf{G} |$.
	This process can be written in terms of the above-mentioned generating
	processes:
	\begin{equation}
\begin{tikzpicture}
	\begin{pgfonlayer}{nodelayer}
		\node [style=none] (39) at (0, 1) {};
		\node [style=small box] (40) at (0, 0) {$F$};
		\node [style=none] (43) at (0, -1) {};
	\end{pgfonlayer}
	\begin{pgfonlayer}{edgelayer}
		\draw [qWire, in=-90, out=90] (43.center) to (40);
		\draw [qWire, in=90, out=-90] (39.center) to (40);
	\end{pgfonlayer}
\end{tikzpicture}
} \ \ = \ \
\InputIfFileExists{Diagrams/PositivityProof2.tikz}{}{\input{./figures/Diagrams/PositivityProof2.tikz}} \,,
	\end{equation}
	where we do not specify the internal structure of the dashed box as the
	actual
	compositional structure of $F$ has no generic specification but we
	assume it is
	a diagram consisting of generating processes.

	We will now show that this  box associated to the process $F$  can
	always be
	rewritten into a diagram which only involves processes in $\mathbf{G}$.
	This
	follows from the fact the we can rewrite any diagram using only
	generating
	processes.

	Suppose, for example, that the diagram involves the process
	$\eta_1^{\Lambda}$,
	that is:
	\begin{equation}
\InputIfFileExists{Diagrams/PositivityProof2.tikz}{}{\input{./figures/Diagrams/PositivityProof2.tikz}}\ \ = \ \
\InputIfFileExists{Diagrams/PositivityProof3.tikz}{}{\input{./figures/Diagrams/PositivityProof3.tikz}} \,.
	\end{equation}
	Since	$A_1^{\Lambda}$ is not an input to the process $F$ (as we are
	assuming
	that the inputs and outputs  system types are in $| \mathbf{G} |$),
	there
	must be a process in the diagram $Diag'$ for which this system,
	$A_1^{\Lambda}$, is an output. There is only one generating process
	which has
	$A_1^{\Lambda}$ as an output, namely, $\xi^{\Lambda}$. Hence, we can
	write
	diagram $Diag'$ as:
	\begin{equation}
\InputIfFileExists{Diagrams/PositivityProof3.tikz}{}{\input{./figures/Diagrams/PositivityProof3.tikz}} \ \  = \ \
\InputIfFileExists{Diagrams/PositivityProof4.tikz}{}{\input{./figures/Diagrams/PositivityProof4.tikz}} \,.
	\end{equation}
	We also know that none of the $A_i^{\Lambda}$ are outputs of the
	process,
	hence, they must be the input of some process within the diagram
	$Diag''$. For
	each of these there is a single generating process which has
	$A_i^{\Lambda} $
	as an input, namely, $\eta_i^{\Lambda} $. This means we can rewrite the
	diagram
	$Diag''$ as:
	\begin{equation}
		\begin{aligned}
\InputIfFileExists{Diagrams/PositivityProof4.tikz}{}{\input{./figures/Diagrams/PositivityProof4.tikz}} \ \  & = \ \
\InputIfFileExists{Diagrams/PositivityProof5.tikz}{}{\input{./figures/Diagrams/PositivityProof5.tikz}} \,.
		\end{aligned}
	\end{equation}
	The explicitly drawn part of the diagram, however, is now nothing but
	the
	non-signalling channel $ \Lambda$, which is a process that lives in
	$\mathbf{G}$:
	\begin{equation}
		\begin{aligned}
\InputIfFileExists{Diagrams/PositivityProof5.tikz}{}{\input{./figures/Diagrams/PositivityProof5.tikz}} \ \ = \ \
\InputIfFileExists{Diagrams/PositivityProof6.tikz}{}{\input{./figures/Diagrams/PositivityProof6.tikz}} \,.
		\end{aligned}
	\end{equation}
	Hence, we have shown that we can redraw the diagram  associated to $F$
	so as
	not to use the generating process $\eta_1^{\Lambda}$. This argument
	clearly
	also applies to any  other $\eta_i^{\Lambda}$ that may appear in the
	specification of $F$,  and a very minor modification of it applies to
	any
	$\xi^{\Lambda}$.

	This means that any process in $\overline{\mathbf{G}\sqcup
			\boldsymbol{\eta}}$
	whose input and output system have types in $| \mathbf{G} |$
	can always be written in a way that only involves  generating processes
	from
	$\mathbf{G}$ and does not involve any generating processes from
	$\{\eta_i^{\Lambda}, \xi^{\Lambda}\}$. As $\mathbf{G}$ is closed under
	composition, we have therefore shown that any process  with input and
	output
	system types in $| \mathbf{G} |$  is necessarily a  valid process in
	$\mathbf{G}$.
\end{proof}

Notice that, in particular, Lemma \ref{lem:positivity} implies that the theory
$\overline{\mathbf{G}\sqcup \boldsymbol{\eta}}$  that we have defined will make
sensible probabilistic predictions, since the classical systems are valid
systems in $\mathbf{G}$ and any processes with only classical inputs and
outputs is necessarily a stochastic map.

\subsection{Proof of Lemma \ref{lem:causality}}\label{se:pl3}
\textbf{Lemma \ref{lem:causality}}. There is a unique discarding effect for
each system in $\overline{\mathbf{G}\sqcup \boldsymbol{\eta}}$.

\begin{proof}
	Here we show that every generating type has a unique discarding effect,
	as the
	generalisation to composite types is straightforward.

	For  each generating system $i$ from the GPT $\mathbf{G}$,
	Lem.~\ref{lem:positivity} implies that the discarding effect for $i$ is
	itself
	a valid process $\mathbf{G}$. Since $\mathbf{G}$ is causal, this means
	that the
	discarding effect for $i$ in $\overline{\mathbf{G}\sqcup
			\boldsymbol{\eta}}$ is
	unique.

	Now we need to show that the claim also holds for system types beyond
	those
	present in the GPT $\mathbf{G}$, i.e., the systems $\{A_i^{\Lambda}\}$.

	Since all processes of $\overline{\mathbf{G} \sqcup \boldsymbol{\eta}}$
	can
	decomposed in terms of generating processes, we can write a generic
	effect for
	$A_i^\Lambda$ as
	\begin{equation}\label{}
\begin{tikzpicture}
	\begin{pgfonlayer}{nodelayer}
		\node [style=none] (0) at (0, -0.75) {};
		\node [style=none] (1) at (0, 0.5) {};
		\node [style=copoint] (2) at (0, 0.75) {$e$};
		\node [style=label] (3) at (0.5, -0.25) {$A_i^{\Lambda}$};
	\end{pgfonlayer}
	\begin{pgfonlayer}{edgelayer}
		\draw [pWire] (1.center) to (0.center);
	\end{pgfonlayer}
\end{tikzpicture}
}\ \ =\ \
\InputIfFileExists{Diagrams/CausalProof2.tikz}{}{\input{./figures/Diagrams/CausalProof2.tikz}}
	\end{equation}
	where, thanks to Lemma~\ref{se:pl2}, we know that $\tilde{e}$ and
	$\sigma'$ are necessarily in $\mathbf{G}$.
	Then, as there is a unique effect for each system in $\mathbf{G}$ we
	have that

	\begin{equation}
\InputIfFileExists{Diagrams/CausalProof2.tikz}{}{\input{./figures/Diagrams/CausalProof2.tikz}}  \	\  = \ \
\InputIfFileExists{Diagrams/CausalProof2_3.tikz}{}{\input{./figures/Diagrams/CausalProof2_3.tikz}}\ \ =:\ \
\InputIfFileExists{Diagrams/CausalProof2_4.tikz}{}{\input{./figures/Diagrams/CausalProof2_4.tikz}}
	\end{equation}
	Now, using the definition of $\eta_i^\Lambda$ we have that:
	\begin{equation} %
\InputIfFileExists{Diagrams/CausalProof2_4.tikz}{}{\input{./figures/Diagrams/CausalProof2_4.tikz}}\ \ =\ \
\InputIfFileExists{Diagrams/CausalProof2_5.tikz}{}{\input{./figures/Diagrams/CausalProof2_5.tikz}}\ \ =\ \
\InputIfFileExists{Diagrams/CausalProof2_6.tikz}{}{\input{./figures/Diagrams/CausalProof2_6.tikz}}\ \
		= \ \ %
\InputIfFileExists{Diagrams/TerminalityProof4.tikz}{}{\input{./figures/Diagrams/TerminalityProof4.tikz}} \ \ =:\ \
\begin{tikzpicture}
	\begin{pgfonlayer}{nodelayer}
		\node [style=upground] (31) at (0, 0) {};
		\node [style=right label] (35) at (0, -1) {$A_i^{\Lambda}$};
		\node [style=none] (38) at (0, -1) {};
	\end{pgfonlayer}
	\begin{pgfonlayer}{edgelayer}
		\draw [pWire] (31) to (38.center);
	\end{pgfonlayer}
\end{tikzpicture}
}
	\end{equation}

	Hence the extra systems in the enlarged theory still satisfy the
	property of
	having an unique effect to each system type.
\end{proof}

\subsection{Proof of Lemma \ref{lem:tomcauspreserv}}\label{se:pl4}
\textbf{Lemma \ref{lem:tomcauspreserv}}.
i) Any process in $\mathsf{Conv}[\overline{\mathbf{G}\sqcup
			\boldsymbol{\eta}}]$ with only	input and output system
types in $| \mathbf{G}
	|$ is a valid process in $ \mathbf{G}$.  ii) There is a unique
discarding
effect for each system in  $\mathsf{Conv}[\overline{\mathbf{G}\sqcup
			\boldsymbol{\eta}}]$.

\begin{proof} \quad \\
	\begin{itemize}
		\item[i)]  From Lem.~\ref{lem:positivity} we know that any
			process in
			$\overline{\mathbf{G}\sqcup \boldsymbol{\eta}}$ with
			only input and output
			system types in $| \mathbf{G} |$ is a valid process in
			$ \mathbf{G}$. Since
			processes in $ \mathbf{G}$ are closed under convex
			combinations, this implies
			that any convex combination of processes in
			$\overline{\mathbf{G}\sqcup
					\boldsymbol{\eta}}$ with only input and
			output system types in $| \mathbf{G} |$
			is a valid process in $ \mathbf{G}$, which proves the
			claim.
		\item[ii)] That there is a unique discarding effect for each
			system immediately
			follows from Lem.~\ref{lem:causality} together with the
			fact that  since there
			is  a unique discarding effect	for each generating
			system type  it is
			impossible to obtain  other discarding effects by
			composition and convex
			combinations.
	\end{itemize}
\end{proof}

\subsection{Proof of Lemma \ref{lem:comp}}\label{se:pl5}
\textbf{Lemma \ref{lem:comp}}. Composition as defined in
Def.~\ref{def:QuotComp} is independent of the choices of representatives. That
is,
\begin{equation}
	\begin{aligned}
\begin{tikzpicture}
	\begin{pgfonlayer}{nodelayer}
		\node [style=small box] (0) at (0, 0) {$f$};
		\node [style=none] (1) at (0, 1) {};
		\node [style=none] (2) at (0, -1) {};
	\end{pgfonlayer}
	\begin{pgfonlayer}{edgelayer}
		\draw [gWire] (2.center) to (0);
		\draw [gWire] (1.center) to (0);
	\end{pgfonlayer}
\end{tikzpicture}
}\	\sim\
\begin{tikzpicture}
	\begin{pgfonlayer}{nodelayer}
		\node [style=small box] (0) at (0, 0) {$f'$};
		\node [style=none] (1) at (0, 1) {};
		\node [style=none] (2) at (0, -1) {};
	\end{pgfonlayer}
	\begin{pgfonlayer}{edgelayer}
		\draw [gWire] (2.center) to (0);
		\draw [gWire] (1.center) to (0);
	\end{pgfonlayer}
\end{tikzpicture}
}\quad                 &
		\text{and}\quad
\begin{tikzpicture}
	\begin{pgfonlayer}{nodelayer}
		\node [style=small box] (0) at (0, 0) {$g$};
		\node [style=none] (1) at (0, 1) {};
		\node [style=none] (2) at (0, -1) {};
	\end{pgfonlayer}
	\begin{pgfonlayer}{edgelayer}
		\draw [gWire] (2.center) to (0);
		\draw [gWire] (1.center) to (0);
	\end{pgfonlayer}
\end{tikzpicture}
}\	\sim \
\begin{tikzpicture}
	\begin{pgfonlayer}{nodelayer}
		\node [style=small box] (0) at (0, 0) {$g'$};
		\node [style=none] (1) at (0, 1) {};
		\node [style=none] (2) at (0, -1) {};
	\end{pgfonlayer}
	\begin{pgfonlayer}{edgelayer}
		\draw [gWire] (2.center) to (0);
		\draw [gWire] (1.center) to (0);
	\end{pgfonlayer}
\end{tikzpicture}
}\qquad
		\implies  \qquad %
\InputIfFileExists{Diagrams/DPProof1.tikz}{}{\input{./figures/Diagrams/DPProof1.tikz}} \	\sim \
\InputIfFileExists{Diagrams/DPProof2.tikz}{}{\input{./figures/Diagrams/DPProof2.tikz}} \quad \text{and} \quad &
\InputIfFileExists{Diagrams/DPProof3.tikz}{}{\input{./figures/Diagrams/DPProof3.tikz}} \  \sim\
\InputIfFileExists{Diagrams/DPProof4.tikz}{}{\input{./figures/Diagrams/DPProof4.tikz}}.
	\end{aligned}
\end{equation}

\begin{proof}
	We start by rewriting the assumptions of the theorem using the
	definition of
	equivalence:
	\begin{equation}\label{ffp}
\begin{tikzpicture}
	\begin{pgfonlayer}{nodelayer}
		\node [style=small box] (0) at (0, 0) {$f$};
		\node [style=none] (1) at (0, 1) {};
		\node [style=none] (2) at (0, -1) {};
	\end{pgfonlayer}
	\begin{pgfonlayer}{edgelayer}
		\draw [gWire] (2.center) to (0);
		\draw [gWire] (1.center) to (0);
	\end{pgfonlayer}
\end{tikzpicture}
}\ \ \sim\ \
\begin{tikzpicture}
	\begin{pgfonlayer}{nodelayer}
		\node [style=small box] (0) at (0, 0) {$f'$};
		\node [style=none] (1) at (0, 1) {};
		\node [style=none] (2) at (0, -1) {};
	\end{pgfonlayer}
	\begin{pgfonlayer}{edgelayer}
		\draw [gWire] (2.center) to (0);
		\draw [gWire] (1.center) to (0);
	\end{pgfonlayer}
\end{tikzpicture}
}\qquad \iff\qquad \forall \tau
		\quad
\InputIfFileExists{Diagrams/CompProof3.tikz}{}{\input{./figures/Diagrams/CompProof3.tikz}}\ \  =\ \
\InputIfFileExists{Diagrams/CompProof4.tikz}{}{\input{./figures/Diagrams/CompProof4.tikz}}
	\end{equation}
	and
	\begin{equation}\label{ggp}
\begin{tikzpicture}
	\begin{pgfonlayer}{nodelayer}
		\node [style=small box] (0) at (0, 0) {$g$};
		\node [style=none] (1) at (0, 1) {};
		\node [style=none] (2) at (0, -1) {};
	\end{pgfonlayer}
	\begin{pgfonlayer}{edgelayer}
		\draw [gWire] (2.center) to (0);
		\draw [gWire] (1.center) to (0);
	\end{pgfonlayer}
\end{tikzpicture}
}\ \ \sim\ \
\begin{tikzpicture}
	\begin{pgfonlayer}{nodelayer}
		\node [style=small box] (0) at (0, 0) {$g'$};
		\node [style=none] (1) at (0, 1) {};
		\node [style=none] (2) at (0, -1) {};
	\end{pgfonlayer}
	\begin{pgfonlayer}{edgelayer}
		\draw [gWire] (2.center) to (0);
		\draw [gWire] (1.center) to (0);
	\end{pgfonlayer}
\end{tikzpicture}
}\qquad \iff\qquad \forall \tau
		\quad
\InputIfFileExists{Diagrams/CompProof7.tikz}{}{\input{./figures/Diagrams/CompProof7.tikz}}\ \ =\ \
\InputIfFileExists{Diagrams/CompProof8.tikz}{}{\input{./figures/Diagrams/CompProof8.tikz}} \,.
	\end{equation}

	Hence, for an arbitrary $\tau$, the following holds:
	\begin{equation}\label{}
		\begin{aligned}
\InputIfFileExists{Diagrams/CompProof9.tikz}{}{\input{./figures/Diagrams/CompProof9.tikz}}\ \  & =:\ \
\InputIfFileExists{Diagrams/CompProof10.tikz}{}{\input{./figures/Diagrams/CompProof10.tikz}}\ \
			\stackrel{\eqref{ffp}}{=}\ \
\InputIfFileExists{Diagrams/CompProof11.tikz}{}{\input{./figures/Diagrams/CompProof11.tikz}}\ \ =\ \
\InputIfFileExists{Diagrams/CompProof12.tikz}{}{\input{./figures/Diagrams/CompProof12.tikz}}\ \ =                     \\
			                                  & \hspace{-0.5cm}=\ \
\InputIfFileExists{Diagrams/CompProof13.tikz}{}{\input{./figures/Diagrams/CompProof13.tikz}}\ \
			=:\ \ %
\InputIfFileExists{Diagrams/CompProof14.tikz}{}{\input{./figures/Diagrams/CompProof14.tikz}}\ \
			\stackrel{\eqref{ggp}}{=}\ \
\InputIfFileExists{Diagrams/CompProof15.tikz}{}{\input{./figures/Diagrams/CompProof15.tikz}}\ \ =\ \
\InputIfFileExists{Diagrams/CompProof16.tikz}{}{\input{./figures/Diagrams/CompProof16.tikz}} ,
		\end{aligned} \,.
	\end{equation}
	This implies that
	\begin{equation}\label{}
\InputIfFileExists{Diagrams/CompProof17.tikz}{}{\input{./figures/Diagrams/CompProof17.tikz}}\ \ \sim\ \
\InputIfFileExists{Diagrams/CompProof18.tikz}{}{\input{./figures/Diagrams/CompProof18.tikz}}\,,
	\end{equation}
	hence proving the first part of the lemma.

	The proof of the second part of the lemma follows in a similar way: for
	an arbitrary $ \tau$,
	\begin{equation}\label{}
		\begin{aligned}
\InputIfFileExists{Diagrams/CompProof19.tikz}{}{\input{./figures/Diagrams/CompProof19.tikz}}\ \  & =:\ \
\InputIfFileExists{Diagrams/CompProof20.tikz}{}{\input{./figures/Diagrams/CompProof20.tikz}}\ \ \stackrel{
				\eqref{ffp}}{=}\ \
\InputIfFileExists{Diagrams/CompProof21.tikz}{}{\input{./figures/Diagrams/CompProof21.tikz}} \ \ =\ \
\InputIfFileExists{Diagrams/CompProof22.tikz}{}{\input{./figures/Diagrams/CompProof22.tikz}}\ \ =\ \
\InputIfFileExists{Diagrams/CompProof23.tikz}{}{\input{./figures/Diagrams/CompProof23.tikz}}             \\   & =: \ \
\InputIfFileExists{Diagrams/CompProof24.tikz}{}{\input{./figures/Diagrams/CompProof24.tikz}}	\ \
			    \stackrel{\eqref{ggp}}{=} \ \
\InputIfFileExists{Diagrams/CompProof25.tikz}{}{\input{./figures/Diagrams/CompProof25.tikz}} \ \ =\ \
\InputIfFileExists{Diagrams/CompProof26.tikz}{}{\input{./figures/Diagrams/CompProof26.tikz}}	\ \ =\ \
\InputIfFileExists{Diagrams/CompProof27.tikz}{}{\input{./figures/Diagrams/CompProof27.tikz}} \,.
		\end{aligned}
	\end{equation}
	Hence,
	\begin{equation}\label{}
\InputIfFileExists{Diagrams/CompProof28.tikz}{}{\input{./figures/Diagrams/CompProof28.tikz}}\ \ \sim\ \
\InputIfFileExists{Diagrams/CompProof29.tikz}{}{\input{./figures/Diagrams/CompProof29.tikz}},
	\end{equation}
	which completes the proof.
\end{proof}

\subsection{Proof of Lemma \ref{lem:convmix}}\label{se:pl6}
\textbf{Lemma \ref{lem:convmix}}. Convex mixtures as defined in
Def.~\ref{def:QuotMix} are independent of the choice of representative. That
is:
\begin{equation}
	\begin{aligned}
\begin{tikzpicture}
	\begin{pgfonlayer}{nodelayer}
		\node [style=small box] (0) at (0, 0) {$f$};
		\node [style=none] (1) at (0, 1) {};
		\node [style=none] (2) at (0, -1) {};
	\end{pgfonlayer}
	\begin{pgfonlayer}{edgelayer}
		\draw [gWire] (2.center) to (0);
		\draw [gWire] (1.center) to (0);
	\end{pgfonlayer}
\end{tikzpicture}
}\ \ \sim\ \
\begin{tikzpicture}
	\begin{pgfonlayer}{nodelayer}
		\node [style=small box] (0) at (0, 0) {$f'$};
		\node [style=none] (1) at (0, 1) {};
		\node [style=none] (2) at (0, -1) {};
	\end{pgfonlayer}
	\begin{pgfonlayer}{edgelayer}
		\draw [gWire] (2.center) to (0);
		\draw [gWire] (1.center) to (0);
	\end{pgfonlayer}
\end{tikzpicture}
}\quad \text{and} & \quad
\begin{tikzpicture}
	\begin{pgfonlayer}{nodelayer}
		\node [style=small box] (0) at (0, 0) {$g$};
		\node [style=none] (1) at (0, 1) {};
		\node [style=none] (2) at (0, -1) {};
	\end{pgfonlayer}
	\begin{pgfonlayer}{edgelayer}
		\draw [gWire] (2.center) to (0);
		\draw [gWire] (1.center) to (0);
	\end{pgfonlayer}
\end{tikzpicture}
}  \ \ \sim  \ \
\begin{tikzpicture}
	\begin{pgfonlayer}{nodelayer}
		\node [style=small box] (0) at (0, 0) {$g'$};
		\node [style=none] (1) at (0, 1) {};
		\node [style=none] (2) at (0, -1) {};
	\end{pgfonlayer}
	\begin{pgfonlayer}{edgelayer}
		\draw [gWire] (2.center) to (0);
		\draw [gWire] (1.center) to (0);
	\end{pgfonlayer}
\end{tikzpicture}
}\quad
		\implies \quad %
\begin{tikzpicture}
	\begin{pgfonlayer}{nodelayer}
		\node [style=small box] (0) at (0, 0) {$pf+(1-p)g$};
		\node [style=none] (1) at (0, 1) {};
		\node [style=none] (2) at (0, -1) {};
	\end{pgfonlayer}
	\begin{pgfonlayer}{edgelayer}
		\draw [gWire] (2.center) to (0);
		\draw [gWire] (1.center) to (0);
	\end{pgfonlayer}
\end{tikzpicture}
} & \sim \ \
\begin{tikzpicture}
	\begin{pgfonlayer}{nodelayer}
		\node [style=small box] (0) at (0, 0) {$pf'+(1-p)g'$};
		\node [style=none] (1) at (0, 1) {};
		\node [style=none] (2) at (0, -1) {};
	\end{pgfonlayer}
	\begin{pgfonlayer}{edgelayer}
		\draw [gWire] (2.center) to (0);
		\draw [gWire] (1.center) to (0);
	\end{pgfonlayer}
\end{tikzpicture}
}.
	\end{aligned}
\end{equation}

\begin{proof}
	The assumptions of the Lemma can be equivalently stated as
	\begin{equation}
\InputIfFileExists{Diagrams/ConvProof5.tikz}{}{\input{./figures/Diagrams/ConvProof5.tikz}}\ \ = \ \ %
\InputIfFileExists{Diagrams/ConvProof6.tikz}{}{\input{./figures/Diagrams/ConvProof6.tikz}}
		\quad
		\forall \, \tau \,,
	\end{equation}
	and
	\begin{equation}
\InputIfFileExists{Diagrams/ConvProof7.tikz}{}{\input{./figures/Diagrams/ConvProof7.tikz}}\ \ = \ \ %
\InputIfFileExists{Diagrams/ConvProof8.tikz}{}{\input{./figures/Diagrams/ConvProof8.tikz}}
		\quad
		\forall \, \tau \,.
	\end{equation}
	In particular, this means that $\forall \, \tau$
	\begin{equation}
		\hspace{-0.5cm}p \ %
\InputIfFileExists{Diagrams/ConvProof5.tikz}{}{\input{./figures/Diagrams/ConvProof5.tikz}}+(1-p)\
\InputIfFileExists{Diagrams/ConvProof7.tikz}{}{\input{./figures/Diagrams/ConvProof7.tikz}}\ \  =\ \	p\
\InputIfFileExists{Diagrams/ConvProof6.tikz}{}{\input{./figures/Diagrams/ConvProof6.tikz}}+(1-p)\
\InputIfFileExists{Diagrams/ConvProof8.tikz}{}{\input{./figures/Diagrams/ConvProof8.tikz}} \,.
	\end{equation}
	Linearity of $\tau$ implies
	\begin{equation}
\InputIfFileExists{Diagrams/ConvProof9.tikz}{}{\input{./figures/Diagrams/ConvProof9.tikz}}\ \  =\ \	%
\InputIfFileExists{Diagrams/ConvProof10.tikz}{}{\input{./figures/Diagrams/ConvProof10.tikz}}
	\end{equation}
	for all $\tau$, which proves the claim.
\end{proof}

\end{document}